\documentclass[11pt]{article}

\usepackage{graphicx, mathtools, amssymb, latexsym, amsmath, amsfonts, amsthm, bbm, tikz}
\usepackage{hyperref}
\hypersetup{
    colorlinks=true, 
    linkcolor=blue, 
    urlcolor=red, 
    linktoc=all 
}

\usepackage{array, fancyhdr, bm, listings, soul, xcolor,titlesec, cancel,ifthen}

\newcommand\eps{\varepsilon}

\DeclareMathOperator*\poly{poly}
\newcommand\one{\textbf{1}}
\newcommand{\floor}[1]{\lfloor {#1} \rfloor}
\newcommand{\ceil}[1]{\lceil {#1}\rceil}
\newcommand\defeq{\ensuremath{\stackrel{\rm def}{=}}} 
\DeclareMathOperator*\rtd{rt-dist}
\DeclareMathOperator*\dist{dist}
\newcommand{\abs}[1]{\left\vert#1\right\vert}
\DeclareMathOperator*{\argmax}{arg\,max}

\theoremstyle{plain}
\newtheorem{theorem}{Theorem}[section]
\newtheorem{lemma}[theorem]{Lemma}
\newtheorem{claim}[theorem]{Claim}

\newtheorem{hypothesis}[theorem]{Hypothesis}
\newtheorem*{lemma*}{Lemma}
\newtheorem*{claim*}{Claim}
\newtheorem*{proposition*}{Proposition}
\newtheorem*{fact*}{Fact}
\newtheorem*{corollary*}{Corollary}
\newtheorem*{hint*}{Hint}

\theoremstyle{definition}
\newtheorem{definition}[theorem]{Definition}
\newtheorem{remark}[theorem]{Remark}

\newtheorem*{theorem*}{Theorem}
\newtheorem*{definition*}{Definition}
\newtheorem*{remark*}{Remark}
\newtheorem*{notation*}{Notation}
\newtheorem*{example*}{Example}
\newtheorem*{examples*}{Examples}
\newtheorem*{question*}{Question}
\newtheorem*{answer*}{Answer}
\newtheorem*{problem*}{Problem}
\newtheorem*{solution*}{Solution}
\newtheorem*{idea*}{Idea}
\newtheorem*{conjecture*}{Conjecture}

\usepackage{pgfplots}
\mathtoolsset{showonlyrefs}
\hypersetup{citecolor=black}
\begin{document}
\title{On Diameter Approximation in Directed Graphs}
\author{Amir Abboud\thanks{Weizmann Institute of Science, \url{amir.abboud@weizmann.ac.il}. This project has received funding from the European Research Council (ERC) under the European Union’s Horizon Europe research and innovation programme (grant agreement No 101078482). Additionally, Amir Abboud is supported by an Alon scholarship and a research grant from the Center for New Scientists at the Weizmann Institute of Science.} , 
Mina Dalirrooyfard\thanks{Massachusetts Institute of Technology, \url{minad@mit.edu}. Partially supported by an Akamai Fellowship.} , 
Ray Li\thanks{UC Berkeley, \url{rayyli@berkeley.edu}. Supported by the NSF Mathematical Sciences Postdoctoral Research Fellowships Program under Grant DMS-2203067, and a UC Berkeley Initiative for Computational Transformation award.} , 
Virginia Vassilevska-Williams\thanks{Massachusetts Institute of Technology, \url{virgi@mit.edu}. Partially supported by the National Science Foundation Grant CCF-2129139.}}
\date{}
\maketitle

\begin{abstract}
Computing the diameter of a graph, i.e. the largest distance, is a fundamental problem that is central in fine-grained complexity. In undirected graphs, the Strong Exponential Time Hypothesis (SETH) yields a lower bound on the time vs. approximation trade-off that is quite close to the upper bounds.
    
In \emph{directed} graphs, however, where only some of the upper bounds apply, much larger gaps remain. Since $d(u,v)$ may not be the same as $d(v,u)$, there are multiple ways to define the problem, the two most natural being the \emph{(one-way) diameter} ($\max_{(u,v)} d(u,v)$) and the \emph{roundtrip diameter} ($\max_{u,v} d(u,v)+d(v,u)$). In this paper we make progress on the outstanding open question for each of them.

    \begin{itemize}
        \item We design the first algorithm for diameter in sparse directed graphs to achieve $n^{1.5-\eps}$ time with an approximation factor better than $2$. The new upper bound trade-off makes the directed case appear more similar to the undirected case. Notably, this is the first algorithm for diameter in sparse graphs that benefits from fast matrix multiplication.

        \item We design new hardness reductions separating roundtrip diameter from directed and undirected diameter. In particular, a $1.5$-approximation in subquadratic time would refute the All-Nodes $k$-Cycle hypothesis, and any $(2-\eps)$-approximation would imply a breakthrough algorithm for approximate $\ell_{\infty}$-Closest-Pair.  Notably, these are the first conditional lower bounds for diameter that are not based on SETH.
    \end{itemize}
    
\end{abstract}

\thispagestyle{empty}
\newpage

\setcounter{page}{1}

\section{Introduction}

The diameter of the graph is the largest shortest paths distance. A very well-studied parameter with many practical applications (e.g. \cite{diam-prac2,diam-prac3,diam-prac4,diam-prac6}), its computation and approximation are also among the most interesting problems in Fine-Grained Complexity (FGC).
Much effort has gone into understanding the approximation vs. running time tradeoff for this problem (see the survey \cite{rubinstein2019seth} and the progress after it \cite{Bonnet21,Bonnet21ic,Li20,Li21,DW21,DLV21}).

Throughout this introduction we will consider $n$-vertex and $m$-edge graphs that, for simplicity, are \emph{unweighted} and \emph{sparse} with $m=n^{1+o(1)}$ edges\footnote{Notably, however, our algorithmic results hold for general graphs, and our hardness results hold even for very sparse graphs.}.
The diameter is easily computable in $\tilde{O}(mn)=n^{2+o(1)}$ time\footnote{The notation $\tilde{O}(f(n))$ denotes $O(f(n)\poly\log(f(n))$.} by computing All-Pairs Shortest Paths (APSP). One of the first and simplest results in FGC \cite{RV13,Williams05} is that any $O(n^{2-\eps})$ time algorithm for $\eps>0$ for the exact computation of the diameter would refute the well-established Strong Exponential Time Hypothesis (SETH) \cite{ipz1,cip10}. 
Substantial progress has been achieved in the last several years \cite{RV13,ChechikLRSTW14,Bonnet21,Bonnet21ic,Li20,Li21,DW21,DLV21}, culminating in an approximation/running time lower bound tradeoff based on SETH, showing that even for undirected sparse graphs, for every $k\geq 2$, there is no $2-1/k-\delta$-approximation algorithm running in $\tilde{O}(n^{1+1/(k-1)-\eps})$ time for some $\delta,\eps>0$.

In terms of upper bounds, the following three algorithms work for both undirected and directed graphs: 
\begin{enumerate}
\item compute APSP and take the maximum distance, giving an exact answer in $\tilde{O}(n^2)$ time, 

\item compute single-source shortest paths from/to an arbitrary node and return the largest distance found, giving a $2$-approximation in $\tilde{O}(n)$ time, and 
\item an algorithm by \cite{RV13,ChechikLRSTW14} giving a $3/2$-approximation in $\tilde{O}(n^{1.5})$ time.\end{enumerate}

For undirected graphs, there are some additional algorithms, given by Cairo, Grossi and Rizzi \cite{CGR} that qualitatively (but not quantitatively) match the tradeoff suggested by the lower bounds:  for every $k\geq 1$ they obtain an $\tilde{O}(n^{1+1/(k+1)})$ time, almost-$(2-1/2^k)$ approximation algorithm, meaning that there is also a small constant additive error.

\newcommand\figureaxes{
\draw[->] (0.98,1) to (2.05,1);
\draw[->] (1,0.98) to (1,2.05);
\node[anchor=west] at (2.03,1) {Apx};
\node[anchor=south] at (1,2.05) {$k$};
\node[anchor=east, rotate=90] at (0.9,1.75) {$\tilde O(m^k)$ time};
}

\def\omeg{2.37285996}
\pgfmathdeclarefunction{usx}{1}{%
  \pgfmathparse {2-1/(2^(#1+2))}
}
\pgfmathdeclarefunction{usy}{1}{%
  \pgfmathparse {1+(2*((2/(\omeg-1))^#1) - ((\omeg-1)^2)/2) / ( ((2/(\omeg-1))^#1) * (7-\omeg) - (\omeg^2 - 1)/2 )}
}
\hyphenation{un-directed}

\begin{figure}
  \begin{center}
    \begin{tikzpicture}[xscale=6, yscale=6, font=\scriptsize]
      \coordinate (ub0) at (1,2);
      \coordinate (ub1) at (1.5,1.5);
      \coordinate (ub2) at (2,1);
      \draw[fill=blue!30,draw=none] (2.01,2.01) rectangle (ub0);
      \draw[fill=blue!30,draw=none] (2.01,2.01) rectangle (ub1);
      \draw[fill=blue!30,draw=none] (2.01,2.01) rectangle (ub2);
      \foreach \x in {2,...,8} {
        \coordinate (cgr\x) at ({2-2^(-\x)}, {1+1.0/(\x+1)});
        \draw[fill=blue!30,draw=none] (2,2) rectangle (cgr\x);
        
      }
      \node[draw, fill=blue!30, circle,scale=0.5] (ub0n) at (ub0) {};
      \node[draw, fill=blue!30, circle,scale=0.5] (ub1n) at (ub1) {};
      \node[draw, fill=blue!30, circle,scale=0.5] (ub2n) at (ub2) {};
      \foreach \x in {2,...,6} {
         \node[draw, fill=blue!20, circle,scale=0.5] (ncgr\x) at (cgr\x) {};
      }
      \foreach \x in {1,...,100} {
        \coordinate (lb\x) at ({2-1.0/(\x+1)}, {1+1.0/\x});
        \draw[fill=red!20,draw=none] (1,1) rectangle (lb\x);
      }
      \draw[fill=red!20,draw=none] (1,1) -- (lb100) -- (2,1);
            \node[text width=90] (prev) at (1.4, 1.25) {Undirected diameter hardness (from SETH) \cite{RV13,BRSVW18,Li21,Bonnet21ic,DLV21}};
      \figureaxes
      \draw[blue,thick] ({7/4},1.5)--({7/4},{1+(\omeg+1)/(\omeg+5)});
      \node[draw, fill=blue!30, circle,scale=0.5] at (ub0) {};
      \node[draw, fill=blue!30, circle,scale=0.5] at (ub1) {};
      \node[draw, fill=blue!30, circle,scale=0.5] at (ub2) {};
      \foreach \x in {1,...,15} {
        \node[draw, fill=red!20, circle,scale=0.5] at (lb\x) {};
      }
      
            \node[] (ub0l) at (1.2, 2.1) {APSP};
      \draw[->] (ub0l) to (ub0n);
      \node[] (ub2l) at (2.2, 1.1) {SSSP};
      \draw[->] (ub2l) [out=200,in=50] to (ub2n);
      \node[] (ub1l) at (1.8, 1.8) {\cite{RV13,ChechikLRSTW14}};
      \draw[->] (ub1l) [out=230,in=90] to (ub1n);
      \node (cgrl) at (1.9,1.4) {\cite{CGR}};
         \draw[->] (cgrl) [out=260,in=0] to (ncgr2);   \draw[->] (cgrl) to (ncgr4);
      \draw[->] (cgrl) to (ncgr6);
    \end{tikzpicture}
  \end{center}
  \caption{Undirected diameter algorithms and hardness.}
  \label{fig:undirected}
\end{figure}
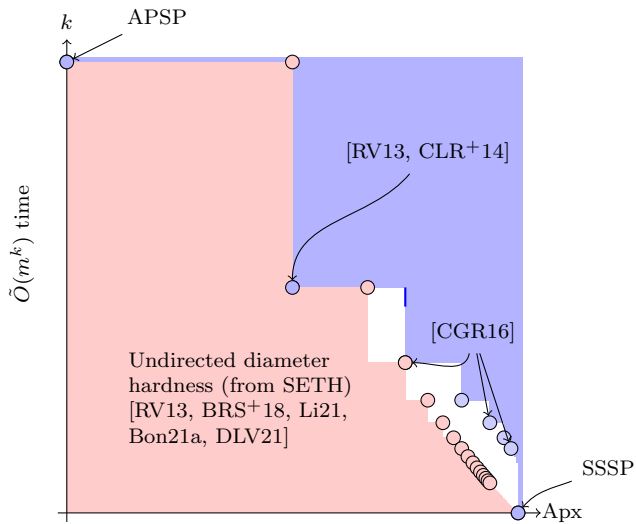

The upper and lower bound tradeoffs for undirected graphs are depicted in Figure~\ref{fig:undirected} ; a gap remains (depicted as white space) because the two trade-offs have different rates.
In directed graphs, however, the gap is significantly larger because an upper bound trade-off is missing (the lower bound tradeoff follows immediately because it is a harder problem).
One could envision for instance, that the conditional lower bounds for directed diameter could be strengthened to show that if one wants a $(2-\eps)$-approximation algorithm, then it must take at least $n^{1.5-o(1)}$ time. 
Since the work of \cite{CGR}, the main open question (also asked by \cite{rubinstein2019seth}) for diameter algorithms in directed graphs has been:

\begin{center}
{\em Why are there only three approximation algorithms for directed diameter, but undirected diameter has an infinite approximation scheme? Is directed diameter truly harder, or can one devise further approximation algorithms for it?} 
\end{center}

\paragraph*{Directed is Closer to Undirected.}
Our first result is that one \emph{can} devise algorithms for directed diameter with truly faster running times than $n^{1.5}$, and approximation ratios between $3/2$ and $2$. 
It turns out that the directed case has an upper bound tradeoff as well, albeit with a worse rate than in the undirected case. 
Conceptually, this brings undirected and directed diameter closer together. See Figure~\ref{fig:directed} for our new algorithms.

\begin{figure}
    \centering
    \begin{tikzpicture}[xscale=6, yscale=6, font=\scriptsize]
      \coordinate (ub0) at (1,2);
      \coordinate (ub1) at (1.5,1.5);
      \coordinate (ub2) at (2,1);
      \draw[fill=blue!30,draw=none] (2.01,2.01) rectangle (ub0);
      \draw[fill=blue!30,draw=none] (2.01,2.01) rectangle (ub1);
      \draw[fill=blue!30,draw=none] (2.01,2.01) rectangle (ub2);
      \foreach \x in {0,...,7} {
        \coordinate (us\x) at ({usx(\x)},{usy(\x)});
        \draw[fill=blue!70,draw=none] (2,1.5) rectangle (us\x);
        \draw[blue] ({usx(\x+1)},{usy(\x)}) -- (us\x) -- ({usx(\x)},{usy(\x-1)});
      }
      \foreach \x in {1,...,100} {
        \coordinate (lb\x) at ({2-1.0/(\x+1)}, {1+1.0/\x});
        \draw[fill=red!20,draw=none] (1,1) rectangle (lb\x);
      }
      \figureaxes

      \node[draw, fill=blue!30, circle,scale=0.5] (ub0n) at (ub0) {};
      \node[draw, fill=blue!30, circle,scale=0.5] (ub1n) at (ub1) {};
      \node[draw, fill=blue!30, circle,scale=0.5] (ub2n) at (ub2) {};
      \foreach \x in {1,...,15} {
        \node[draw, fill=red!30, circle,scale=0.5] at (lb\x) {};
      }
      \foreach \x in {0,...,5} {
        \node[draw, fill=blue!70, circle,scale=0.5] (n\x) at (us\x) {};
      }
      \node[text width=100] (new) at (1.4, 1.25) {Hardness (from SETH) \cite{RV13,BRSVW18,Bonnet21,DW21,Li21}};
      \node[text width=70] (us1l) at (2.25, 1.6) {Our algorithms (Thm.~\ref{thm:alg-main})};
      \draw[->] (us1l) [out=190,in=40] to (n0);
      \draw[->] (us1l) [out=190,in=60] to (n1);
      \draw[->] (us1l) [out=190,in=120] to (n5);
      \node[] (ub0l) at (1.2, 2.1) {APSP};
      \draw[->] (ub0l) to (ub0n);
      \node[] (ub2l) at (2.2, 1.1) {SSSP};
      \draw[->] (ub2l) [out=200,in=50] to (ub2n);
      \node[] (ub1l) at (1.8, 1.8) {\cite{RV13,ChechikLRSTW14}};
      \draw[->] (ub1l) [out=230,in=90] to (ub1n);
      \node[] () at (1.5, 2.17) {\textbf{Directed Diameter}};
    \end{tikzpicture}
    \caption[Directed diameter algorithms and hardness]{Directed diameter algorithms and hardness. All tradeoffs hold for both weighted and unweighted graphs (though citations may differ for weighted vs. unweighted).}
    \label{fig:directed}
\end{figure}
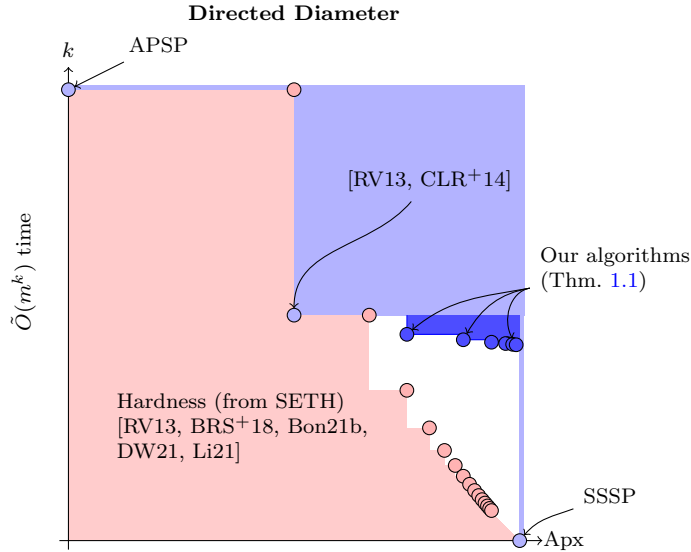

\begin{theorem}
  Let $k=2^{t+2}$ for a nonnegative integer $t\ge 0$.
  For every $\varepsilon>0$ (possibly depending on $m$), there exists a randomized $2-\frac{1}{k}+\varepsilon$-approximation algorithm for the diameter of a directed weighted graphs in time $\tilde O(m^{1+\alpha}/\varepsilon)$, for 
  \begin{align}
  \alpha=\frac{2(\frac{2}{\omega-1})^t - \frac{(\omega-1)^2}{2}}{(\frac{2}{\omega-1})^t(7-\omega) - \frac{\omega^2-1}{2}}.
  \end{align}
\label{thm:alg-main}
\end{theorem}

The constant $2 \leq \omega< 2.37286$ in the theorem refers to the fast matrix multiplication exponent \cite{almanvw21}. A surprising feature of our algorithms is that we utilize fast matrix multiplication techniques to obtain faster algorithms for a problem in sparse graphs. Prior work on shortest paths has often used fast matrix multiplication
to speed-up computations, but to our knowledge, all of this work is for dense graphs (e.g. \cite{agm97,Seidel95,Zwick02,minajenny}).
Breaking the $n^{1.5}$ bound with a combinatorial algorithm is left as an open problem.

\paragraph*{Roundtrip is Harder.}

One unsatisfactory property of the shortest paths distance measure in directed graphs is that it is not symmetric ($d(u,v)\neq d(v,u)$) and is hence not a metric. Another popular distance measure used in directed graphs that is a metric is the roundtrip measure. Here the {\em roundtrip distance} $\tilde{d}(u,v)$ between vertices $u,v$ is $d(u,v)+d(v,u)$.

Roundtrip distances were first studied in the distributed computing community in the 1990s \cite{CowenW99}. In recent years, powerful techniques were developed to handle the fast computation of sparse roundtrip spanners, and approximations of the {\em minimum} roundtrip distance, i.e. the shortest cycle length, the girth, of a directed graph \cite{PachockiRSTW18,ChechikLRS20,DalirrooyfardW20,ChechikL21}. These techniques give hope for new algorithms for the {\em maximum} roundtrip distance, the roundtrip diameter of a directed graph.

Only the first \emph{two} algorithms in the list in the beginning of the introduction work for roundtrip diameter: compute an exact answer by computing APSP, and a linear time 2-approximation that runs SSSP from/to an arbitrary node. These two algorithms work for any distance {\em metric}, and surprisingly there have been no other algorithms developed for roundtrip diameter. The only fine-grained lower bounds for the problem are the ones that follow from the known lower bounds for diameter in undirected graphs, and these cannot explain why there are no known subquadratic time algorithms that achieve a better than $2$-approximation.

\begin{center}{\em Are there $O(n^{2-\eps})$ time algorithms for roundtrip diameter in sparse graphs that achieve a $2-\delta$-approximation for constants $\eps,\delta>0$?}\end{center}

This question was considered e.g. by \cite{AbboudWW16} who were able to obtain a hardness result for the related roundtrip radius problem, showing that under a popular hypothesis, such an algorithm for roundtrip radius does not exist.
One of the main questions studied at the ``Fine-Grained Approximation Algorithms and Complexity Workshop" at Bertinoro in 2019 was to obtain new algorithms or hardness results for roundtrip diameter. Unfortunately, however, no significant progress was made, on either front.

The main approach to obtaining hardness for roundtrip diameter, was to start from the Orthogonal Vectors (OV) problem and reduce it to a gap version of roundtrip diameter, similar to all known reductions to (other kinds of) diameter approximation hardness. 
Unfortunately, it has been difficult to obtain a reduction from OV to roundtrip diameter that has a larger gap than that for undirected diameter; in Section~\ref{sec:LBoverview} we give some intuition for why this is the case. 

In this paper we circumvent the difficulty by giving stronger hardness results for roundtrip diameter starting from {\em different} problems and hardness hypotheses. We find this intriguing because all previous conditional lower bounds for (all variants of) the diameter problem were based on SETH. In particular, it gives a new approach for resolving the remaining gaps in the undirected case, where higher SETH-based lower bounds are provably impossible (under the so-called NSETH) \cite{Li21}.

Our first negative result conditionally proves that any $5/3-\eps$ approximation for roundtrip requires $n^{2-o(1)}$ time; separating it from the undirected and the directed one-way cases where a $1.5$-approximation in $\tilde{O}(n^{1.5})$ time is possible.
This result is based on a reduction from the so-called All-Nodes $k$-Cycle problem.

\begin{definition}[All-Nodes $k$-Cycle in Directed Graphs]
Given a $k$ partite directed graph $G=(V,E),V=V_1 \cup \cdots \cup V_k,$ whose edges go only between ``adjacent'' parts $E \subseteq \bigcup_{i=1}^k V_i \times V_{i+1 \mod k}$, decide if all nodes $v \in V_1$ are contained in a $k$-cycle in $G$.
\end{definition}

This problem can be solved for all $k$ in time $O(nm)$, e.g. by running an APSP algorithm, and in subquadratic $O(m^{2-1/k})$ for any fixed $k$ \cite{AYZ97}. Breaking the quadratic barrier for super-constant $k$ has been a longstanding open question; we hypothesize that it is impossible.

\begin{hypothesis}
\label{hypo:ankc}
No algorithm can solve the All-Nodes $k$-Cycle problem in sparse directed graphs for all
$k \geq 3$ in $O(n^{2-\delta})$ time, with $\delta>0$.
\end{hypothesis}

Similar hypotheses have been used in recent works \cite{agarwal2018fine,lincolnsoda,ancona2019algorithms,probst2020new}. The main difference is that we require all nodes in $V_1$ to be in cycles; such variants of hardness assumptions that are obtained by changing a quantifier in the definition of the problem are popular, see e.g. \cite{AbboudWW16,BringmannC20,AbboudBHS22}.

\begin{theorem}
  Under Hypothesis~\ref{hypo:ankc}, for all $\varepsilon,\delta>0$, no algorithm can $5/3-\varepsilon$ approximate the roundtrip diameter of a sparse directed unweighted graph in $O(n^{2-\delta})$ time.
\label{thm:lb-ankc-unw}
\end{theorem}

We are thus left with a gap between the linear time factor-$2$ upper bound and the subquadratic factor-$5/3$ lower bound. 
A related problem with a similar situation is the problem of computing the eccentricity of all nodes in an undirected graph \cite{AbboudWW16}; there, $5/3$ is the right number because one can indeed compute a $5/3$-approximation in subquadratic time \cite{ChechikLRSTW14}.
Could it be the same here?

Alas, our final result is a reduction from the following classical problem in geometry to roundtrip diameter, establishing a barrier for any better-than-$2$ approximation in subquadratic time.

\begin{definition}[Approximate $\ell_\infty$ Closest-Pair]
  Let $\alpha>1$. The $\alpha$-approximate $\ell_\infty$ Closest-Pair (CP) problem is, given $n$ vectors $v_1,\dots,v_n$ of some dimension $d$ in $\mathbb{R}^n$, determine if there exists $v_i$ and $v_j$ with $\|v_i-v_j\|_\infty\le 1$, or if for all $v_i$ and $v_j$, $\|v_i-v_j\|_\infty\ge \alpha$.
\end{definition}

Closest-pair problems are well-studied in various metrics; the main question being whether the naive $n^2$ bound can be broken (when $d$ is assumed to be $n^{o(1)}$). For $\ell_{\infty}$ specifically, a simple reduction from OV proves a quadratic lower bound for $(2-\eps)$-approximations \cite{indyk01}; but going beyond this factor with current reduction techniques runs into a well-known ``triangle-inequality'' barrier (see \cite{Rubinstein18,KarthikM20}). This leaves a huge gap from the upper bounds that can only achieve $O(\log \log n)$ approximations in subquadratic time \cite{indyk01}.
Cell-probe lower bounds for the related nearest-neighbors problem suggest that this log-log bound may be optimal \cite{ACP08}; if indeed constant approximations are impossible in subquadratic time then the following theorem implies a tight lower bound for roundtrip diameter.

\begin{theorem}
If for some $\alpha\ge 2, \eps>0$ there is a $2-\frac{1}{\alpha}-\varepsilon$ approximation algorithm in time $O(m^{2-\eps})$ for roundtrip diameter in unweighted graphs, then for some $\delta>0$ there is an $\alpha$-approximation for $\ell_\infty$-Closest-Pair with vectors of dimension $d\le n^{1-\delta}$ in time $\tilde O(n^{2-\delta})$.

\label{thm:lb-linfty-unw}
\end{theorem}

In particular, a $2-\eps$ approximation for roundtrip diameter in subquadratic time implies an $\alpha$-approximation for the $\ell_{\infty}$-Closest-Pair problem in subquadratic time, for some $\alpha=O(1/\eps)$.
Thus, any further progress on the roundtrip diameter problem requires a breakthrough on one of the most basic algorithmic questions regarding the $\ell_{\infty}$ metric (see Figure~\ref{fig:rt}).

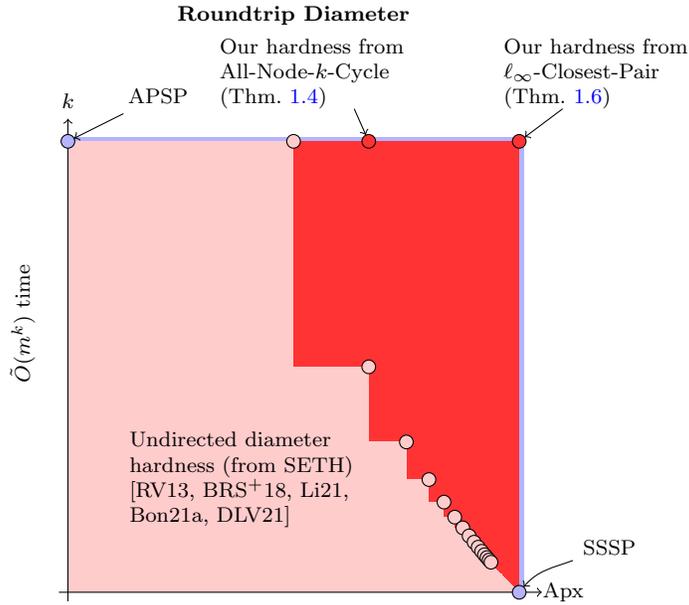
\begin{figure}
    \centering
    \begin{tikzpicture}[xscale=6, yscale=6, font=\scriptsize]
      \coordinate (ub0) at (1,2);
      \coordinate (ub2) at (2,1);
      \draw[fill=blue!30,draw=none] (2.01,2.01) rectangle (ub0);
      \draw[fill=blue!30,draw=none] (2.01,2.01) rectangle (ub2);
      \coordinate (us1) at ({5/3},2);
      \coordinate (us2) at (2,2);
      \draw[fill=red!80,draw=none] (1.5,1) rectangle (us2);
      \draw[fill=red!80,draw=none] (1.5,1) rectangle (us1);
      \foreach \x in {1,...,100} {
        \coordinate (lb\x) at ({2-1.0/(\x+1)}, {1+1.0/\x});
        \draw[fill=red!20,draw=none] (1,1) rectangle (lb\x);
      }
      \figureaxes

      \node[draw, fill=blue!30, circle,scale=0.5] (ub0n) at (ub0) {};
      \node[draw, fill=blue!30, circle,scale=0.5] (ub2n) at (ub2) {};
      \foreach \x in {1,...,15} {
        \node[draw, fill=red!20, circle,scale=0.5] at (lb\x) {};
      }
      \node[draw, fill=red!80, circle,scale=0.5] (n1) at (us1) {};
      \node[draw, fill=red!80, circle,scale=0.5] (n2) at (us2) {};
      \node[text width=90] (prev) at (1.4, 1.25) {Undirected diameter hardness (from SETH) \cite{RV13,BRSVW18,Li21,Bonnet21ic,DLV21}};
      \node[text width=90,inner sep=0, outer sep=0] (lb1l) at (1.6, 2.15) {Our hardness from All-Node-$k$-Cycle (Thm.~\ref{thm:lb-ankc-unw})};
      \node[text width=80, inner sep=0, outer sep=0] (lb2l) at (2.2, 2.15) {Our hardness from $\ell_\infty$-Closest-Pair (Thm.~\ref{thm:lb-linfty-unw})};
      \draw[->] (lb1l) -- (n1);
      \draw[->] (lb2l) -- (n2);

      \node[] (ub0l) at (1.2, 2.1) {APSP};
      \draw[->] (ub0l) to (ub0n);
      \node[] (ub2l) at (2.2, 1.1) {SSSP};
      \draw[->] (ub2l) [out=200,in=50] to (ub2n);
      \node[] () at (1.5, 2.28) {\textbf{Roundtrip Diameter}};
    \end{tikzpicture}
    \caption[Roundtrip Diameter algorithms and hardness]{Roundtrip Diameter algorithms and hardness. All tradeoffs hold for both weighted and unweighted graphs (though citations may differ for weighted vs. unweighted). The previously best hardness results were those inherited from undirected diameter.}
    \label{fig:rt}
\end{figure}

\subsection{Related Work}

Besides the diameter and the roundtrip diameter, there is another natural version of the diameter problem in directed graphs called Min-Diameter \cite{AbboudWW16,dalirrooyfard2019approximation,minajenny}.
The distance between $u,v$ is defined as the $\min(d(u,v),d(v,u))$.\footnote{Note that the Max-Diameter version where we take the max rather than the min is equal to the one-way version.}
This problem seems to be even harder than roundtrip because even a $2$-approximation in subquadratic time is not known.

The fine-grained complexity results on diameter (in the sequential setting) have had interesting consequences for computing the diameter in distributed settings (specifically in the CONGEST model). 
Techniques from both the approximation algorithms and from the hardness reductions have been utilized, see e.g. \cite{PelegRT12,AbboudCKP21,AnconaCDEW20}.
It would be interesting to explore the consequences of our techniques on the intriguing gaps in that context \cite{GKP20}.

\subsection{Organization}

In the main body of the paper, we highlight the key ideas in our main results (Theorem~\ref{thm:alg-main} and Theorem~\ref{thm:lb-linfty-unw}) by proving an ``easy version'' of each theorem, and in the appendices, we establish the full theorems.
First, we establish some preliminaries in Section~\ref{sec:prelims}.
In Section~\ref{sec:alg-74}, we prove the special case of Theorem~\ref{thm:alg-main} when $t=0$, giving a $7/4$-approximation of the diameter in directed unweighted graphs in time $O(m^{1.458})$. 
In Section~\ref{sec:alg-main} of the appendix, we generalize this proof to all parameters $t\ge 0$ and to weighted graphs.
In Section~\ref{sec:LBoverview} we give an overview of the hardness reductions in this paper.
In Section~\ref{sec:linfty}, we prove a weakening of Theorem~\ref{thm:lb-linfty-unw} that only holds for weighted graphs. 
Later, in Section~\ref{sec:linfty-unw} of the appendix, we extend this lower bound to unweighted graphs.
In Section~\ref{sec:ankc}, we prove a weakening of Theorem~\ref{thm:lb-ankc-unw}: under Hypothesis~\ref{hypo:ankc}, there is no $5/3-\varepsilon$ approximation of roundtrip diameter in weighted graphs.
Later, in Section~\ref{sec:ankc-unw}, we extend this lower bound to unweighted graphs.

\section{Preliminaries}
\label{sec:prelims}

All logs are base $e$ unless otherwise specified.
For reals $a\ge 0$, let $[\pm a]$ denote the real interval $[-a,a]$.
For a boolean statement $\varphi$, let $\one[\varphi]$ be 1 if $\varphi$ is true and 0 otherwise.

For a vertex $v$ in a graph, let $\deg(v)$ denote its degree.
For $r\ge 0$, let $B_r^{in}(v) = \{u: d(u,v) \le r\}$ be the in-ball of radius $r$ around $v$, and let $B_r^{out}(v) = \{u: d(v,u) \le r\}$ be the out-ball of radius $r$ around $v$.
For $r\ge 0$, let $B_r^{in+}(v)$ be $B_r^{in}(v)$ and their in-neighbors, and let $B_r^{out+}(v)$ be $B_r^{out}(v)$ and their out-neighbors.

Throughout, let $\omega\le 2.3728596$ denote the matrix multiplication constant.
We use the following lemma which says that we can multiply sparse matrices quickly.
\begin{lemma}[see e.g. Theorem 2.5 of \cite{KaplanSV06}]
  We can multiply a $a\times b$ and a $b\times a$ matrix, each with at most $ac$ nonzero entries, in time $O(ac\cdot a^{\frac{\omega-1}{2}})$.\footnote{In \cite{KaplanSV06}, this runtime of $O(ac\cdot a^{\frac{\omega-1}{2}})$ is stated only for the case $ac>a^{(\omega+1)/2}$. However, the runtime bound for this case works for other cases as well so the lemma is correct for all matrices.}
\label{lem:matmul}
\end{lemma}

We repeatedly use the following standard fact.
\begin{lemma}
  Given two sets $B\subset V$ with $B$ of size $k$ and $V$ of size $2m$, a set of $4(m/k)\log m$ uniformly random elements of $V$ contains an element of $B$ with probability at least $1-\frac{1}{m^2}$.
\label{lem:hit}
\end{lemma}
\begin{proof}
  The probability that $B$ is not hit is $(1-\frac{k}{2m})^{4m/k\log m} \le e^{-2\log m} = \frac{1}{m^2}$.
\end{proof}

\section{$7/4$-approximation of directed (one-way) diameter}
\label{sec:alg-74}
In this section, we prove Theorem~\ref{thm:alg-main} in the special case of $t=0$ and unweighted graphs.
That is, we give a $7/4$-approximation of the (one-way) diameter of a directed unweighted graph in $O(m^{1.4575})$ time. 
For the rest of this section, let $\alpha=\frac{\omega+1}{\omega+5}\le 0.4575$.

Before stating the algorithm and proof, we highlight how our algorithm differs from the undirected algorithm of \cite{CGR}. At a very high level, all known diameter approximation algorithms compute some pairs of distances, and use the triangle inequality to infer other distances, saving runtime. Approximating diameter in directed graphs is harder than in undirected graphs because distances are not symmetric, so we can only use the triangle inequality ``one way." For example, we always have $d(x,y)+d(y,z) \ge d(x,z)$, but not necessarily $d(x,y)+d(z,y) \ge d(x,z)$. The undirected algorithm \cite{CGR} crucially uses the triangle inequality ``both ways," so it was not clear whether their algorithm could be adapted to the directed case. We get around this barrier using matrix multiplication together with the triangle inequality to infer distances quickly. We consider the use of matrix multiplication particularly interesting because, previously, matrix multiplication had only been used for diameter in dense graphs, but we leverage it in sparse graphs. 

\begin{theorem}
  Let $\alpha=\frac{\omega+1}{\omega+5}$. There exists a randomized $7/4$-approximation algorithm for the diameter of an unweighted directed graph running in $\tilde O(m^{1+\alpha})$ time.
\label{thm:alg-74}
\end{theorem}
\begin{proof}
  It suffices to show that, for any positive integer $D > 0$, there exists an algorithm $\mathcal{A}_D$ running in time $\tilde O(m^{1+\alpha})$ that takes as input any graph and accepts if the diameter is at least $D$, rejects if the diameter is less than $4D/7$, and returns arbitrarily otherwise.
  Then, we can find the diameter up to a factor of $7/4$ by running binary search with $\mathcal{A}_D$,\footnote{
  We have to be careful not to lose a small additive factor. Here are the details: Let $D^*$ be the true diameter. Initialize $hi = n, lo = 0$. Repeat until $hi-lo=1$: let $mid=\floor{(hi+lo)/2}$, run $\mathcal{A}_{mid}$, if accept, set $lo=mid$, else $hi=mid$. One can check that $hi \ge D^*+1$ and $lo \le 7D^*/4$ always hold. If we return $lo$ after the loop breaks, the output is always in $[D^*, 7D^*/4]$.}
  which at most adds a factor of $O(\log n)$.

  We now describe the algorithm $\mathcal{A}_D$.
  The last two steps, illustrated in Figure~\ref{fig:74-1} contain the key new ideas. 
  
\begin{enumerate}
  \item First, we apply a standard trick that replaces the input graph on $n$ vertices and $m$ edges with an $2m$-vertex graph of max-degree-3 that preserves the diameter: replace each vertex $v$ with a $\deg(v)$-vertex cycle of weight-0 edges and where the edges to $v$ now connect to distinct vertices of the cycle. From now on, we work with this max-degree-3 graph on $2m$ vertices. 

  \item\label{74:step:1} Sample $4m^{\alpha}\log m$ uniformly random vertices and compute each vertex's in- and out-eccentricity. If any such vertex has (in- or out-) eccentricity at least $4D/7$ Accept.

  \item\label{74:step:2} For every vertex $v$, determine if $|B_{D/7}^{out}(v)|\le m^{\alpha}$. If such a vertex $v$ exists, determine if any vertex in $B_{D/7}^{out+}(v)$ has eccentricity at least $4D/7$, and Accept if so.

  \item\label{74:step:3} For every vertex $v$, determine if $|B_{D/7}^{in}(v)|\le m^{\alpha}$. If such a vertex $v$ exists, determine if any vertex in $B_{D/7}^{in+}(v)$ has eccentricity at least $4D/7$, and Accept if so.

  \item\label{74:step:4} Sample $4m^{1-\alpha}\log m$ uniformly random vertices $\hat S$. Let $S^{out}=\{s\in\hat S: |B^{out}_{2D/7}(s)|\le m^{1-\alpha}\}$ and $S^{in}=\{s\in \hat S: |B^{in}_{2D/7}(s)|\le m^{1-\alpha}\}$.
  Compute $B^{out}_{2D/7}(s)$ and $B^{out+}_{2D/7}(s)$ for $s\in S^{out}$, and $B^{in}_{2D/7}(s)$ and $B^{in+}_{2D/7}(s)$ for $s\in S^{in}$.

  \item\label{74:step:5} Let $A^{out}\in \mathbb{R}^{S^{out}\times V}$ be the $|S^{out}|\times n$ matrix where $A_{s,v}=\one[v\in B^{out}_{2D/7}(s)]$. Let $A^{in}\in \mathbb{R}^{V\times S^{in}}$ be the $n\times |S^{in}|$ matrix where $A^{in}_{v,s}=\one[v\in B^{in}_{2D/7}(s)]$ if $\floor{4D/7}=2\floor{2D/7}$ and $A^{in}_{v,s}=\one[v\in B^{in+}_{2D/7}(s)]$ otherwise. Compute $A^{out}\cdot A^{in}\in\mathbb{R}^{S^{out}\times S^{in}}$ using sparse matrix multiplication. If the product has any zero entries, Accept, otherwise Reject.
\end{enumerate}

\begin{figure}
  \begin{center}
    \includegraphics[width=0.9\textwidth]{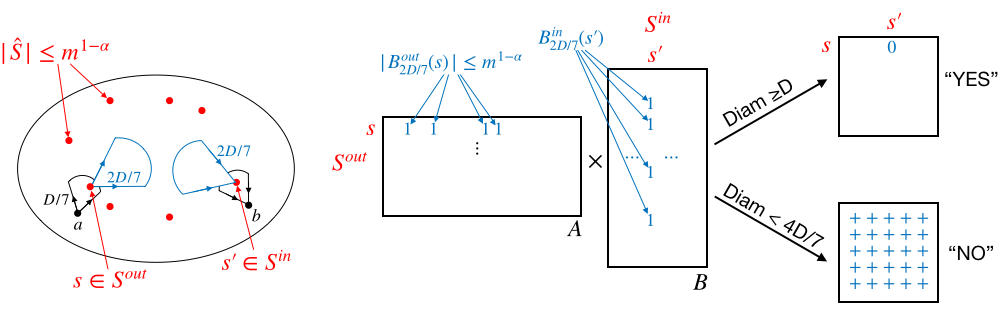}
  \end{center}
  \caption{Steps~\ref{74:step:4} and \ref{74:step:5}. If $d(a,b)\ge D$ and Steps~\ref{74:step:1}, \ref{74:step:2}, and \ref{74:step:3} do not accept, with high probability, set $\hat S$ hits the $D/7$ out- and in- neighborhoods of $a$ and $b$ at vertices $s$ and $s'$, respectively, that must have distance at least $5D/7$ by the triangle inequality. Thus, checking all pairs of distances in $S^{out}\times S^{in}$, which can be done quickly with sparse matrix multiplication, distinguishes at Step~\ref{74:step:5} whether the diameter is at least $D$ or less than $4D/7$.}
  \label{fig:74-1}
\end{figure}

  \paragraph*{Runtime.} 
  Computing a single eccentricity takes time $O(m)$, so Step~\ref{74:step:1} takes time $\tilde O(m^{1+\alpha})$. 
  For Step~\ref{74:step:2} checking if $|B^{out}_{D/7}(v)|\le m^{\alpha}$ takes $O(m^\alpha)$ time for each $v$ via a partial Breadth-First-Search (BFS). Here we use that the max-degree is 3.
  If $|B^{out}_{D/7}(v)|\le m^{\alpha}$, there are at most $3m^\alpha$ eccentricity computations which takes time $O(m^{1+\alpha})$.
  Step~\ref{74:step:3} takes time $O(m^{1+\alpha})$ for the same reason.
  Similarly, we can complete Step~\ref{74:step:4} by running partial BFS for each $s\in \hat S$ until $m^{1-\alpha}$ vertices are visited.
  This gives $S^{out}$ and $S^{in}$ and also gives $B^{out}_{2D/7}(s)$ and $B^{out+}_{2D/7}(s)$ for $s\in S^{out}$ and $B^{in}_{2D/7}(s)$ and $B^{in+}_{2D/7}(s)$ for $s\in S^{in}$.
  For Step~\ref{74:step:5}, the runtime is the time to multiplying sparse matrices. Matrix $A^{out}$ has at most $|\hat S|\le 4m^{1-\alpha}\log m$ rows each with at most $\max_{s\in S^{out}}|B^{out}_{2D/7}(s)|\le m^{1-\alpha}$ entries, and similarly $A^{in}$ has at most $4m^{1-\alpha}\log m$ columns each with at most $\max_{s\in S^{in}}|B^{in+}_{2D/7}(s)| \le 3m^{1-\alpha}$ entries.
  The sparse matrix multiplication takes time $\tilde O(m^{(2-2\alpha)}\cdot m^{(1-\alpha)\frac{\omega-1}{2}})=\tilde{O}(m^{1+\alpha})$ by Lemma~\ref{lem:matmul} with $a=m^{1-\alpha}, b=n, c=m^{1-\alpha}$.

  \paragraph*{If the Diameter is less than $4D/7$, we always reject.}
  Clearly every vertex has eccentricity less than $4D/7$, so we indeed do not accept at Steps~\ref{74:step:1}, \ref{74:step:2}, and \ref{74:step:3}.
  In Step~\ref{74:step:4}, we claim for every $s\in S^{out}, s'\in S^{in}$ there exists $v$ such that $A_{s,v}^{out}=A_{v,s'}^{in}=1$, so that $(A^{out}\cdot A^{in})_{s,s'}\ge 1$ for all $s\in S^{out}$ and $s'\in S^{in}$ and thus we reject.
  Fix $s\in S^{out}$ and $s'\in S^{in}$. By the diameter bound, $d(s,s') \le \floor{4D/7}$. 
  Let $v$ be the last vertex on the $s$-to-$s'$ shortest path such that $d(s,v)\le\floor{2D/7}$, and, if it exists, let $v'$ be the vertex after $v$.
  Clearly $A^{out}_{s,v} = 1$. We show $A^{in}_{v,s'}=1$ as well.
  If $v=s'$, then clearly $v\in B^{in}_{2D/7}(s')$ so $A^{in}_{v,s'}=1$ as desired.
  Otherwise $d(s,v) = \floor{2D/7}$.
  If $\floor{4D/7}=2\floor{2D/7}$, then $d(v,s')\le d(s,s')-d(s,v) \le \floor{4D/7}-\floor{2D/7} = \floor{2D/7}$, so $v\in B^{in}_{2D/7}(s')$ and $A^{in}_{v,s'}=1$, so again $A^{in}_{v,s'}=1$.
  If $\floor{4D/7}=2\floor{2D/7}+1$, then $d(v',s')\le d(s,s')-d(s,v') \le \floor{4D/7}-(\floor{2D/7}+1) = \floor{2D/7}$, so $v'\in B^{in}_{2D/7}(s')$ and thus $v\in B^{in+}_{2D/7}(s')$ and $A^{in}_{v,s'}=1$, as desired.
  This covers all cases, so we've shown we reject.

  \paragraph*{If the Diameter is at least $D$, we accept with high probability.}
  Let $a$ and $b$ be vertices with $d(a,b)\ge D$.

  If $|B^{out}_{3D/7}(a)|> m^{1-\alpha}$, Step~\ref{74:step:1} computes the eccentricity of some $v\in B^{out}_{3D/7}(a)$ with high probability (by Lemma~\ref{lem:hit}), which is at least $d(v,b)\ge d(a,b)-d(a,v)\ge 4D/7$ by the triangle inequality, so we accept.
  Similarly, we accept with high probability if $|B^{in}_{3D/7}(b)|> m^{1-\alpha}$.
  Thus we may assume that $|B^{out}_{3D/7}(a)|, |B^{in}_{3D/7}(b)|\le m^{1-\alpha}$ for the rest of the proof. 

  If $|B^{out}_{D/7}(v)|\le m^{\alpha}$ for any vertex $v$, then either (i) $d(v,b)\ge 4D/7$, in which case $v$ has eccentricity at least $4D/7$ and we accept at Step~\ref{74:step:2}, or (ii) $d(v,b)\le 4D/7$, in which case there is a vertex $u\in B^{out+}_{D/7}(v)$ on the $v$-to-$b$ path with  $d(u,b)\le 3D/7$ (take the $u\in B^{out+}_{D/7}(v)$ closest to $b$ on the path). Then $d(a,u)\ge 4D/7$ by the triangle inequality and we accept in Step~\ref{74:step:2} as we perform a BFS from $u$.
  Thus we may assume $|B^{out}_{D/7}(v)|> m^\alpha$ for all vertices $v$.
  Similarly, because of Step~\ref{74:step:3}, we may assume $|B^{in}_{D/7}(v)|>m^\alpha$ for all vertices $v$.
  
  In particular, we may assume $|B^{out}_{D/7}(a)|> m^\alpha$ and $|B^{in}_{D/7}(b)|>m^{\alpha}$.
  Figure~\ref{fig:74-1} illustrates this last step.
  Then $\hat S$ hits $B^{out}_{D/7}(a)$ with high probability (by Lemma~\ref{lem:hit}), so $B^{out}_{D/7}(a)$ has some $s\in \hat S$ with high probability, and similarly $B^{in}_{D/7}(b)$ has some $s'\in \hat S$ with high probability.
  The triangle inequality implies that $B^{out}_{2D/7}(s) \subset B^{out}_{3D/7}(a)$, so $|B^{out}_{2D/7}(s)|\le |B^{out}_{3D/7}(a)|\le m^{1-\alpha}$ and thus $s\in S^{out}$. Similarly $s'\in S^{in}$.
  By the triangle inequality, we have $d(s,s')\ge d(a,b) - d(a,s) - d(s',b) \ge D-D/7-D/7 = 5D/7$.
  Then we must have $(A\cdot B)_{s,s'} = 0$, as otherwise there is a $v$ such that $d(s,v)\le \floor{2D/7}$ and $d(v,s')\le 4D/7-\floor{2D/7}$, contradicting $d(s,s')\ge 5D/7$.
  Hence, we accept at step 5, as desired.
\end{proof}

\section{Hardness Reductions for Roundtrip}

\subsection{Overview}
\label{sec:LBoverview}

In this paper we prove hardness results for roundtrip diameter that go beyond the $2$ vs. $3$ barrier. Before presenting the proofs, let us begin with an abstract discussion on why this barrier arises and (at a high level) how we overcome it.

All previous hardness results for diameter are by reductions from OV (or its generalization to multiple sets).
In OV, one is given two sets of vectors of size $n$ and dimension $d=\poly\log n$, $A$ and $B$, and one needs to determine whether there are $a\in A, b\in B$ that are orthogonal. SETH implies that OV requires $n^{2-o(1)}$ time \cite{Williams05}. In a reduction from OV to a problem like diameter, one typically has nodes representing the vectors in $A$ and $B$, as well as nodes $C$ representing the coordinates, and if there is an orthogonal vector pair $a,b$, then the corresponding nodes in the diameter graph are far (distance $\geq 3$), and otherwise {\em all} pairs of nodes are close (distance $\leq 2$). 
Going beyond the $2$ vs. $3$ gap is difficult because each node $a \in A$ must have distance $\leq 2$ to each coordinate node in $C$, regardless of the existence of an orthogonal pair, and then it is automatically at distance $2+1$ from any node $b \in B$ because each $b$ has at least one neighbor in $C$. So even if $a,b$ are orthogonal, the distance will not be more than $3$.

The key trick for proving a higher lower bound (say $3$ vs. $5$) for roundtrip is to have two sets of coordinate nodes, a $C^{fwd}$ set that can be used to go \emph{forward} from $A$ to $B$, and a $C^{bwd}$ set that can be used to go back.
The default roundtrip paths from $A$/$B$ to each of these two sets will have different forms, and this asymmetry will allow us to overcome the above issue. 
This is inspired by the difficulty that one faces when trying to make the subquadratic $3/2$-approximation algorithms for undirected and directed diameter work for roundtrip.

Unfortunately, there is another (related) issue when reducing from OV. First notice that all nodes within $A$ and within $B$ must always have small distance (or else the diameter would be large). This can be accomplished simply by adding direct edges of weight $1.5$ between all pairs (within $A$ and within $B$); but this creates a dense graph and makes the quadratic lower bound uninteresting. Instead, such reductions typically add auxiliary nodes to simulate the $n^2$ edges more cheaply, e.g. a star node $o$ that is connected to all of $A$. But then the node $o$ must have small distance to $B$, decreasing all distances between $A$ and $B$. 

Overcoming this issue by a similar trick seems impossible. Instead, our two hardness results bypass it in different ways. 

The reduction from $\ell_{\infty}$-Closest-Pair starts from a problem that is defined over \emph{one set of vectors} $A$ (not two) which means that the coordinates are ``in charge'' of connecting all pairs within $A$. We remark that while OV can also be defined over one set (monochromatic) instead of two (bichromatic) and that it remains SETH hard; that would prevent us from applying the above trick of having a forward and a backward sets of coordinate nodes. Our reduction in Section~\ref{sec:linfty} is able to utilize the structure of the metric in order to make both ideas work simultaneously. 

The reduction from All-Node $k$-Cycle relies on a different idea: it uses a construction where only a small set of $n$ pairs $a_i \in A, b_i \in B$ are ``interesting'' in the sense that we do not care about the distances for other pairs (in order to solve the starting problem).
Then the goal becomes to connect all pairs within $A$ and within $B$ by short paths, without decreasing the distance for the $(a_i,b_i)$ pairs.
A trick similar to the \emph{bit-gadget} \cite{DBLP:conf/soda/AbboudGW15,AbboudCKP21} does the job, see Section~\ref{sec:ankc} of the appendix. For the complete reduction see Section~\ref{sec:ankc-unw}.

\subsection{Weighted Roundtrip $2-\varepsilon$ hardness from $\ell_\infty$-CP}
\label{sec:linfty}
In this section, we highlight the key ideas in Theorem~\ref{thm:lb-linfty-unw} by proving a weaker version, showing the lower bound for weighted graphs.
We extend the proof to unweighted graphs in Section~\ref{sec:linfty-unw}.

The main technical lemma is showing that to $\alpha$-approximate $\ell_\infty$-Closest-Pair, it suffices to do so on instances where all vector coordinates are in $[\pm(0.5+\varepsilon)\alpha]$.
Towards this goal, we make the following definition.
\begin{definition}
  The  $\alpha$-approximate \emph{$\beta$-bounded} $\ell_\infty$-Closest-Pair problem is, given $n$ vectors $v_1,\dots,v_n$ of dimension $d$ in $[-\beta,\beta]^d$ determine if there exists $v_i$ and $v_j$ with $\|v_i-v_j\|_\infty\le 1$, or if for all $v_i$ and $v_j$, $\|v_i-v_j\|_\infty\ge \alpha$.
\end{definition}

We now prove the main technical lemma.
\begin{lemma}
  Let $\varepsilon\in(0,1/2)$ and $\alpha>1$.
  If one can solve $\alpha$-approximate $(0.5+\varepsilon)\alpha$-bounded $\ell_\infty$-CP on dimension $O(d\varepsilon^{-1}\log n)$ in time $T$, then one can solve $\alpha$-approximate $\ell_\infty$-CP on dimension $d$ in time $T+O_\varepsilon(dn\log n)$, where in $O_\varepsilon(\cdot)$ we neglect dependencies on $\varepsilon$. 
\label{lem:bounded}
\end{lemma}
\begin{proof}
  Start with an $\ell_\infty$ instance $\Phi=(v_1,\dots,v_n)$.
  We show how to construct a bounded $\ell_\infty$ instance $\Phi'$ such that $\Phi$ has two vectors with $\ell_\infty$ distance $\le 1$ if and only if $\Phi'$ has two vectors with $\ell_\infty$ distance $\le 1$.

  First we show we may assume that $v_1,\dots,v_n$ are on domain $[0,\alpha n]$.
  Suppose that $x\in[d]$.
  Reindex $v_1,\dots,v_n$ in increasing order of $v_i[x]$ (by sorting).
  Let $v_1',\dots,v_n'$ be vectors identical to $v_1,\dots,v_n$ except in coordinate $x$, where instead
  \begin{align}
    v_i'[x] = \sum_{j=0}^{i-1} \min(\alpha, v_{j+1}[x]-v_j[x])
  \end{align}
  for $i=1,\dots,n$, where the empty sum is 0.
  We have that $v_i'[x]\le \alpha n$ for all $i$, and furthermore $|v_i'[x]-v_j'[x]|\ge \alpha$ if and only if $|v_i[x]-v_j[x]|\ge \alpha$ and also $|v_i'[x]-v_j'[x]|\le 1$ if and only if $|v_i[x]-v_j[x]|\le 1$.
  Hence, the instance given by $v_1',\dots,v_n'$ is a YES instance if and only if the instance $\Phi$ is a YES instance, and is a NO instance if and only if the instance $\Phi$ is a NO instance.
  Repeating this with all other coordinates $x$ gives an instance $\Phi'$ such that $\Phi'$ is a YES instance if and only if $\Phi$ is a YES instance, and $\Phi'$ is a NO instance if and only if $\Phi'$ is a NO instance, and furthermore $\Phi'$ has vectors on $[0,\alpha n]$.

  Now we show how to construct an $\ell_\infty$-CP instance in dimension $O_\varepsilon(d\log n)$ vectors with coordinates in $[\pm(0.5+\varepsilon)\alpha]$.
  \begin{lemma}
    Let $\varepsilon\in(0,0.5)$ and $\alpha > 1$. 
    For any real number $M$, there exists two maps $g:[0,M] \to [-(0.5+\varepsilon)\alpha,(0.5+\varepsilon)\alpha]^{2\ceil{\varepsilon^{-1}}+1}$ and $h:[0,M]\to [0, M/2]$ such that for all $a,b\in[0,M]$, we have $\min(|a-b|,\alpha) = \min(\|(g(a),h(a))-(g(b),h(b))\|_\infty,\alpha)$. (here, $(g(\cdot),h(\cdot))$ is a length $2\ceil{\varepsilon^{-1}}+2$ vector.)
    Furthermore, $g$ and $h$ can be computed in $O_\varepsilon(1)$ time.
    \label{lem:bounded-1.5}
  \end{lemma}
  \begin{proof}
     It suffices to consider when $\varepsilon^{-1}$ is an integer.
     Let $f_{z}:\mathbb{R}\to [-(0.5+\varepsilon)\alpha,(0.5+\varepsilon)\alpha]$ be the piecewise function
    \begin{align}
      f_{z}(x) \ &= \ 
      \left\{
      \begin{tabular}{ll}
      $-(0.5+\varepsilon)\alpha$ & if $x\le z-(0.5+\varepsilon)\alpha$\\
      $(0.5+\varepsilon)\alpha$ & if $x\ge z+(0.5+\varepsilon)\alpha$\\
      $x-z$ & otherwise
      \end{tabular}
      \right.
    \end{align}
    For $a\in[M]$, define $g(a)\in \mathbb{R}^{2\varepsilon^{-1}+1}$ and $h(a)\in\mathbb{R}$ as follows, where we index coordinates by $-\varepsilon^{-1},\dots,-1,0,1,\varepsilon^{-1}$ for convenience 
    \begin{align}
      g(a)_i &= f_{M/2 + 0.5i\varepsilon \alpha}(a) \text{ for }-\varepsilon^{-1}\le i\le \varepsilon^{-1} \nonumber\\
      h(a) &= |a-M/2|.
    \end{align}
    Clearly $g$ and $h$ have the correct codomain, and they can be computed in $O_\varepsilon(1)$ time.
    Additionally, note that $f_z(x)$ and $|x-M/2|$ are 1-Lipschitz functions of $x$ for all $z$, so $g$ is a Lipschitz function and thus $\|g(a)-g(b)\|_\infty\le |a-b|$.

    Now, it suffices to show that $\min(\|(g(a),h(a))-(g(b),h(b))\|_\infty,\alpha)\ge \min(|a-b|,\alpha)$.
    If $a$ and $b$ are on the same side of $M/2$, then $\|h(a)-h(b)\|_\infty\ge ||a-M/2| -|b-M/2|| = |a-b|$, as desired.
    Now suppose $a$ and $b$ are on opposite sides of $M/2$, and without loss of generality $a < M/2 < b$.
    Let $0\le i\le \varepsilon^{-1}$ be the largest integer such that $a \le M/2 - i\varepsilon\alpha$ ($i=0$ works so $i$ always exists).
    If $i=\varepsilon^{-1}$, then $a < M/2 - \alpha$ and 
    \begin{align}
      \|g(a)-g(b)\|_\infty
      \ge f_{M/2-0.5\alpha}(b) - f_{M/2-0.5\alpha}(a) 
      \ge 0.5\alpha - (-0.5\alpha)
      = \alpha 
      \ge  \min(|a-b|,\alpha),
    \end{align}
    as desired.
    Now assume $i<\varepsilon^{-1}$.
    Let $z = M/2 + (0.5 - i\varepsilon) \alpha$.
    By maximality of $i$, we have $a-z\in[-(0.5+\varepsilon)\alpha,-0.5\alpha]$.
    We have $g(\cdot)_{\varepsilon^{-1} - 2i} = f_z(\cdot)$ by definition of $g$.
    By the definition of $f_z(\cdot)$, since $a\in[z-(0.5+\varepsilon)\alpha,z-0.5\alpha]$ and $b\ge a$, we have $\min(f_z(b)-f_z(a),\alpha) = \min(b-a,\alpha)$.
    Thus, 
    \begin{align}
      \min(\|g(a)-g(b)\|_\infty,\alpha)
      &\ge \min\left(g(b)_{\varepsilon^{-1}-2i}-g(a)_{\varepsilon^{-1}-2i},\alpha\right) \nonumber\\
      &= \min(f_z(b)-f_z(a),\alpha)
      = \min(b-a,\alpha),
    \end{align}
    as desired.
    In either case, we have $\min(\|g(a)-g(b)\|_\infty,\alpha)\ge \min(|a-b|,\alpha)$, so we conclude that $\min(\|g(a)-g(b)\|_\infty,\alpha)= \min(|a-b|,\alpha)$
  \end{proof}
  Iterating Lemma~\ref{lem:bounded-1.5} gives the following.
  \begin{lemma}
    Let $\varepsilon\in(0,1/2)$.
    There exists a map $g:[0,\alpha n] \to [\pm(0.5+\varepsilon)\alpha]^{4\ceil{\varepsilon^{-1}}\log n}$ such that for all $a,b\in[0,\alpha n]$, we have $\min(|a-b|,\alpha) = \min(\|g(a)-g(b)\|_\infty,\alpha)$.
    Furthermore, $g$ can be computed in $O_\varepsilon(\log n)$ time.
    \label{lem:bounded-2}
  \end{lemma}
  \begin{proof}
    For $\ell=1,\dots$, let $M_\ell=\alpha n/2^{\ell-1}$, and let $g_\ell^*:[M_\ell]\to[\pm(0.5+\varepsilon)\alpha]^{2\ceil{\varepsilon^{-1}}+1}$ and $h_\ell^*:[M_\ell]\to [M_{\ell+1}]$ be the functions given by Lemma~\ref{lem:bounded-1.5}.
    For $\ell=0,1,\dots$, let $g_\ell:[0,\alpha n]\to [-(0.5+\varepsilon)\alpha,(0.5+\varepsilon)\alpha]^{\ell(2\ceil{\varepsilon^{-1}}+1)}$ and $h_\ell:[0,\alpha n] \to [0,\alpha n/2^\ell]$ be such that $g_0(x)=()$ is an empty vector, $h_0(x) = x$ is the identity, and for $\ell\ge 1$, $g_{\ell}(x) = (g_{\ell-1}(x), g_\ell^*(h_{\ell-1}(x)))$ and $h_\ell(x) = h_\ell^*(h_{\ell-1}(x))$.
    By Lemma~\ref{lem:bounded-1.5}, we have that 
    \begin{align}
      \label{eq:bounded-2}
      &\min \left(\|(g_{\ell-1}(a),h_{\ell-1}(a))-(g_{\ell-1}(b),h_{\ell-1}(b))\|_\infty,\alpha\right)  \nonumber\\
      &= \min\left(\|\big(g_{\ell-1}(a),g_{\ell}^*(h_{\ell-1}(a)), h_{\ell}^*(h_{\ell-1}(a))\big)-\big(g_{\ell}(b), g_{\ell}^*(h_{\ell-1}(b)), h_{\ell}^*(h_{\ell-1}(b))\big)\|_\infty,\alpha\right)  \nonumber\\
      &= \min\left(\|\big(g_{\ell}(a),h_{\ell}(a)\big)-\big(g_{\ell}(b),h_{\ell}(b)\big)\|_\infty,\alpha\right) 
    \end{align}
    for all $\ell$.
    For $\ell=\ceil{\log n}$, the vector $g(a) \defeq (g_\ell(a),h_\ell(a)-0.5\alpha)$ has every coordinate in $[\pm(0.5+\varepsilon)\alpha]$, and by \eqref{eq:bounded-2}, we have 
    \begin{align}
      \min(|a-b|,\alpha) 
      &= \min(|g_0(a)-g_0(b)|,\alpha)  \nonumber\\
      &= \min(|g_\ell(a)-g_\ell(b)|,\alpha) 
      = \min(|g(a)-g(b)|_\infty,\alpha),
    \end{align}
    as desired.
    The length of this vector is at most $\ceil{\log n}(2\ceil{\varepsilon^{-1}}+1) + 1$, which we bound by $4\ceil{\varepsilon^{-1}}\log n$ for simplicity (and pad the corresponding vectors with zeros).
  \end{proof}
  To finish, let $g:[0,\alpha n]\to [\pm (0.5+\varepsilon)\alpha]$ be given by Lemma~\ref{lem:bounded-2}, and let the original $\ell_\infty$ instance be $v_1,\dots,v_n$.
  Let the new $(0.5+\varepsilon)\alpha$-bounded $\ell_\infty$ instance be $w_i = (g(v_i[x]))_{x\in [d]}$ of length $4d\ceil{\varepsilon^{-1}}\log n$.
\end{proof}

We now prove our goal for this section, Theorem~\ref{thm:lb-linfty-unw} for weighted graphs.
\begin{theorem}
If for some $\alpha\ge 2, \eps>0$ there is a $2-\frac{1}{\alpha}-\varepsilon$ approximation algorithm in time $O(m^{2-\eps})$ for roundtrip diameter in \emph{weighted} graphs, then for some $\delta>0$ there is an $\alpha$-approximation for $\ell_\infty$-Closest-Pair with vectors of dimension $d\le n^{1-\delta}$ in time $\tilde O(n^{2-\delta})$.
\end{theorem}
\begin{proof}
  By Lemma~\ref{lem:bounded} it suffices to prove that there exists an $O(n^{2-\delta})$ time algorithm for $\alpha$-approximate $(0.5+\varepsilon)\alpha$-bounded $\ell_\infty$-CP for $\varepsilon=(4\alpha)^{-1}$.

  \begin{figure}
    \begin{center}
    \includegraphics[width=0.8\textwidth]{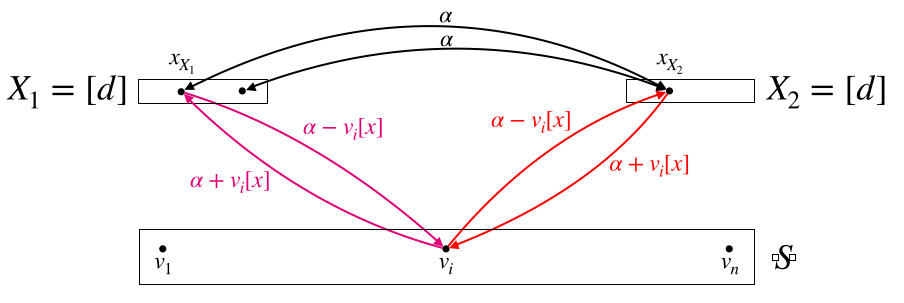}
    \end{center}
    \caption{The roundtrip diameter instance $G$ for $\ell_\infty$-CP hardness.}
    \label{fig:ellinfty}
  \end{figure}

  Let $\Phi$ be the bounded-domain $\ell_\infty$-CP instance with vectors $v_1,\dots,v_n\in[\pm(0.5+\varepsilon)\alpha]^n$.
  Then construct a graph $G$ (see Figure~\ref{fig:ellinfty}) with vertex set $S\cup X_1\cup X_2$ where $X_1=X_2=[d]$ and $S=[n]$.
  We identify vertices with the notations $i_S, x_{X_1},$ and $x_{X_2}$, for $i\in[n]$ and $x\in[d]$.
  Draw directed edges
  \begin{enumerate}
  \item from $i_S$ to $x_{X_1}$, of weight $\alpha+v_i[x]$,
  \item from $x_{X_1}$ to $i_S$, of weight $\alpha-v_i[x]$,
  \item from $i_S$ to $x_{X_2}$, of weight $\alpha-v_i[x]$,
  \item from $x_{X_2}$ to $i_S$, of weight $\alpha+v_i[x]$, and
  \item between any two vertices in $X_1\cup X_2$, of weight $\alpha$.
  \end{enumerate}
  Note that all edge weights are nonnegative, and any two vertices in $X_1\cup X_2$ are roundtrip distance $2\alpha$, and any $s\in S$ and $x\in X_1\cup X_2$ are distance $2\alpha$.
  Suppose $\Phi$ has no solution, so that every pair has $\ell_\infty$ distance $\alpha$.
  Then for vertices $i_S,j_S$, there exists a coordinate $x$ such that $v_i[x] - v_j[x]$ is either $\ge \alpha$ or $\le -\alpha$.
  Without loss of generality, we are in the case $v_i[x] - v_j[x]\ge \alpha$.
  Then the path $i_S \to x_{X_2} \to j_S \to x_{X_1} \to i_S$ is a roundtrip path of length
  \begin{align}
    (\alpha - v_i[x]) + (\alpha+v_j[x]) + (\alpha + v_j[x]) + (\alpha-v_i[x])
    = 4\alpha - 2(v_i[x]-v_j[x]) \le 2\alpha.
  \end{align}
  So when $\Phi$ has no solution, the roundrip diameter is at most $2\alpha$.

  On the other hand, suppose $\Phi$ has a solution $i,j$ such that for all $x$, $|v_i[x] - v_j[x]| \le 1$.
  Then, as every edge has weight at least $(0.5-\varepsilon)\alpha$, 
  \begin{align}
    d(i_S,j_S)
    &\ge \min\left(\min_{x\in[d]}\left(d(i_S,x_{X_1})+d(x_{X_1},j_S), d(i_S,x_{X_2})+d(x_{X_2},j_S)\right), 4(0.5-\varepsilon)\alpha\right)\nonumber\\
    &\ge \min\left(\min_{x\in[d]}(\alpha+v_i[x] +\alpha-v_j[x], \alpha+v_j[x] + \alpha-v_i[x]), 2\alpha-4\varepsilon\alpha\right) \nonumber\\
    &\ge \min(2\alpha-1,2\alpha-4\alpha\varepsilon)
    \ = \  2\alpha-1.
  \end{align}
  Similarly, we have 
  \begin{align}
    d(j_S,i_S)\ge 2\alpha-1,
  \end{align}
  so we have
  \begin{align}
    d_{RT}(j_S,i_S)\ge 4\alpha-2.
  \end{align}
  so in this case the RT-diameter is at least $4\alpha-2$.
  A $2-\alpha^{-1}-\varepsilon$ approximation for RT diameter can distinguish between RT diameter $4\alpha-2$ and RT-diameter $2\alpha$.
  Thus, a $2-\alpha-\varepsilon$ approximation for RT diameter solves $\alpha$-approximate $\ell_\infty$-CP.
\end{proof}

\section{Weighted Roundtrip $5/3-\varepsilon$ hardness from All-Nodes $k$-Cycle}
\label{sec:ankc}

In this section, we highlight the key ideas in Theorem~\ref{thm:lb-ankc-unw} by proving a weaker version, showing the lower bound for weighted graphs.
We extend the proof to unweighted graphs in Section~\ref{sec:ankc-unw}.

\begin{theorem}
Under Hypothesis~\ref{hypo:ankc}, for all $\varepsilon,\delta>0$, no algorithm can $5/3-\varepsilon$ approximate the roundtrip diameter of a sparse directed weighted graph in $O(n^{2-\delta})$ time. 
\end{theorem}

\begin{proof}

Let $G=(V,E),V=V_1 \cup \cdots \cup V_k, E \subseteq \bigcup_{i=1}^k V_i \times V_{i+1 \mod k}$ be the input graph to the All-Nodes $k$-Cycle problem.
The reduction constructs a new graph $G'$ as follows. See Figure \ref{fig:weighted_reduction}.

\begin{figure}[!b]
    \centering
    \includegraphics[width=0.6\textwidth]{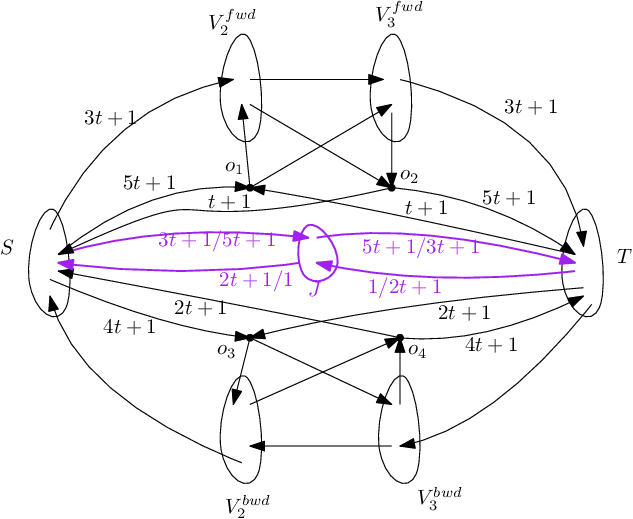}
    \caption{All-Nodes $3$-Cycle to Weighted Roundtrip Diameter approximation reduction. The edges without weights on them have weight 1. There are bidirectional edges of weight $3t$ between any $(o_i,o_j)$ that are not shown in the figure. The purple edges are between a node $S\cup T$ and $J$. An edge of weight $w_1/w_0$ means that if this edge is between $g_i\in J$ and $a\in S\cup T$, then if $\bar{a}[i]=1$ this edge has weight $w_1$ and if $\bar{a}[i]=0$ this edge has weight $w_0$.}
    \label{fig:weighted_reduction}
\end{figure}

\begin{itemize}
\item Each set $V_i$ for $i \in \{ 2,\ldots,k \}$ has two copies in $G'$ one set $V_i^{fwd}$ will be used for interesting \emph{forward} paths and one set $V_i^{bwd}$ that will be used for interesting \emph{backward} paths.
Naturally, the copy of a node $x \in V_i$ in the copy $V_i^{fwd}$ will be denoted $x^{fwd}$ and its copy in $V_i^{bwd}$ will be denoted $x^{bwd}$.
\item The set $V_1$ has two copies that we will call $S$ and $T$. The interesting pairs in our construction will be a subset of the pairs in $S \times T$.
We will use the letters $a,b,c,\ldots$ to denote the nodes in $V_1$. The two copies of a node $a \in V_1$ that are in $S$ and $T$ will be denoted by $a$ and $a'$ such that $a\in S$ and $a' \in T$.
The interesting pairs will in fact be the $n$ pairs $(a,a') \in S \times T$.

\item Let us assume that $|V_1|=n$ and that each node $a\in V_1$ is associated with a unique identifier $\bar{a}$ on $d=O(\log{n})$ bits such that for any pair $a,b \in V_1$ if $a \neq b$ then the two identifiers $\bar{a},\bar{b}$ have at least two coordinates $i,j \in [d]$ where $\bar{a}[i]=1$ while $\bar{b}[i]=0$, and $\bar{a}[j]=0$ while $\bar{b}[j]=1$. In words, we can always find a bit that is $1$ in one but $0$ in the other. In addition, we require that for all $a,b$ there exist two coordinates $i,j \in [d]$ where both $\bar{a}[i]= \bar{b}[i]=1$ and $\bar{a}[i] = \bar{b}[i]=0$, meaning that both are $1$ and $0$. Such identifiers can be obtained, e.g., by taking the bit representation of the name of the node and concatenating it with its complement, then adding a $0$ and a $1$ to all identifiers.

\item There are also some new auxiliary nodes. Most importantly there is a \emph{bit-gadget} comprised of a set $J=\{g_1,\ldots,g_d\}$ of $d$ nodes. In addition, there are four special nodes that connect ``everyone to everyone'' in certain sets; thus, let us denote them $o_1,o_2,o_3,o_4$ where o is for omni.
 
\end{itemize}

The edges of $G'$ and their weights are as follows. Let $t>2k$ be a large enough integer; the reduction will make it difficult to distinguish between diameter $6t \pm O(k)$ and diameter $10t \pm O(k)$. 

\begin{itemize}
\item For each $ i \in \{2,\ldots,k\}$ and for each edge $(x,y) \in V_i \times V_{i+1}$ in $G$, we add two edges to $G'$: one forwards $(x^{fwd},y^{fwd}) \in V_i^{fwd} \times V_{i+1}^{fwd}$ and one backwards $(y^{bwd},x^{bwd}) \in V_{i+1}^{bwd} \times V_{i}^{bwd}$. The weight on these edges is $1$, which can be thought of as negligible because it is $0\cdot t +1$.

\item Each edge leaving $V_1$ in $G$, i.e. an edge $(a,x) \in V_1 \times V_2$, becomes two edges: a forwards $(a,x^{fwd}) \in S \times V_2^{fwd}$ of weight $3\cdot t+1$ and a backwards edge $(x^{bwd},a) \in  V_2^{bwd} \times S$ of weight $0 \cdot t + 1$. 

\item Each edge going into $V_1$ in $G$, i.e. an edge $(x,a) \in V_k \times V_1$, becomes two edges: a forwards $(x^{fwd},a) \in V_k^{fwd} \times T$ of weight $3 \cdot t +1$ and a backwards edge $(a,x^{bwd}) \in T \times V_k^{bwd}$ of weight $0 \cdot t +1$. 

\end{itemize}

The edges defined so far are the main ones. A node $a \in S$ can reach its copy $a' \in T$ with a (forwards) path of weight $6\cdot t +k$ if and only if $a$ is in a $k$-cycle in $G$, and in which case there is also a backwards path of weight $0\cdot t +k$ from $a'\in T$ to $a \in S$. 
This will indeed be the difficult condition to check for an algorithm (under Hypothesis~\ref{hypo:ankc} about the complexity of $k$-cycle) and the rest of the construction aims to make the diameter of $G'$ depend solely on whether this condition is satisfied; and importantly, to make it vary by a large constant factor based on this condition.
This is accomplished with the edges that we define next.

\begin{itemize}

\item The first o-node $o_1$ serves to connect everyone in the set $S$ to everyone in $ V^{fwd}_2 \cup \cdots \cup V^{fwd}_k$ with weight $5 \cdot t+O(1)$. This could have been achieved more simply by having direct edges of weight $5t$ from everyone in $S$ to everyone in those sets. However, this would have incurred $n^2$ edges; the node $o_1$ simulates this with $O(n)$ edges. 
It is connected with edges $(o_1,v)$ to all nodes $v \in V^{fwd}_2 \cup \cdots \cup V^{fwd}_k$. The weight of these edges is $1$. And every node $a \in S$ is connected with an edge of weight $5\cdot t +1$ to $o_1$.

\item At the same time, the node $o_1$ serves to connect everyone in $T$ to everyone in $ V^{fwd}_2 \cup \cdots \cup V^{fwd}_k$ with weight $1 \cdot t+O(1)$. This is achieved by connecting every node $a' \in T$ with an edge $(a',o_1)$ of weight $1 \cdot t +1$ to $o_1$.

\item The second o-node $o_2$ serves to connect everyone in $ V^{fwd}_2 \cup \cdots \cup V^{fwd}_k$ to everyone in $S$ with weight $1 \cdot t +O(1)$.
Every node $v \in V^{fwd}_2 \cup \cdots \cup V^{fwd}_k$ has a direct edge $(v,o_2)$ to $o_2$ with weight $1$, and the node $o_2$ is connected to every node $a \in S$ with an edge $(o_2,a)$ of weight $1\cdot t + 1$.

\item And $o_2$ also serves to connect everyone in $ V^{fwd}_2 \cup \cdots \cup V^{fwd}_k$ to everyone in $T$ with weight $5\cdot t +O(1)$. Thus, we add an edge $(a',o_2)$ of weight $5 \cdot t+1$ for all nodes $a' \in T$. 

\item The third o-node $o_3$ connects everyone in $T$ to everyone in $V^{bwd}_2 \cup \cdots  \cup V^{bwd}_k$ with weight $2t+O(1)$. There are edges of weight $1$ from $o_3$ to all nodes in $V^{bwd}_2 \cup \cdots  \cup V^{bwd}_k$, and there are edges of weight $2 \cdot t +1$ from every node in $T$ to $o_3$.

\item Moreover, $o_3$ connects everyone in $S$ to everyone in $V^{bwd}_2 \cup \cdots  \cup V^{bwd}_k$ with weight $4t+O(1)$. There is an edge of weight $4\cdot t+1$ from every node in $S$ to $o_3$.

\item The fourth and last o-node $o_4$ connects everyone in $V^{bwd}_2 \cup \cdots  \cup V^{bwd}_k$ to everyone in $T$ with weight $4t+O(1)$. There are edges of weight $1$ from every node in $V^{bwd}_2 \cup \cdots  \cup V^{bwd}_k$ to $o_4$, and there are edges of weight $4 \cdot t+1$ from $o_4$ to every node in $T$.

\item Similarly, $o_4$ connects everyone in $V^{bwd}_2 \cup \cdots  \cup V^{bwd}_k$ to everyone in $S$ with weight $2t +O(1)$. There are edges of weight $2 \cdot t+1$ from $o_4$ to every node in $S$.

\item There are bi-directional edges of weight $3 \cdot t$ between all pairs of nodes in $\{o_1,o_2,o_3,o_4\}$.

\end{itemize}

At this point, our construction is nearly complete. Almost all pairs of nodes have a roundtrip of cost $6t+O(1)$, and a node $a \in V_1$ that does not appear in a $k$-cycle in $G$ causes the pair $(a,a') \in S \times T$ to have a roundtrip distance of at least $10\cdot t$.
However, we still have to worry about the pairs within $S$ (and also within $T$); currently their roundtrip distance to each other is $\geq 8t$ even if we are in a YES instance of the $k$-cycle problem. 
The next and final gadget $J$, the \emph{bit-gadget}, will make all distances within $S$ and within $T$ at most $6t+O(1)$ \emph{without making the interesting pairs $(a,a') \in S \times T$ closer than $10t$.}
Unfortunately, we do not know how to achieve the latter guarantee when the set of interesting pairs is larger than $O(n)$. If we could make the roundtrip distances within $S$ smaller without decreasing the roundtrips to $T$ for \emph{all pairs} in $S \times T$ we could have a similar lower bound under SETH rather than Hypothesis~\ref{hypo:ankc}.
The edges that make up the bit-gadget are as follows.

\begin{itemize}

\item Every node $a \in S$ is connected to and from every node $g_j$ in $J$, but the weights on the edges vary based on the identifier $\bar{a}$. 
For a coordinate $j \in [d]$, let $\bar{a}[j] \in \{0,1\}$ be the $j^{th}$ bit in the identifier $\bar{a}$.
\begin{itemize}
\item If $\bar{a}[j] = 1$ we set the weight of the edge $(a,g_j)$ to $3 \cdot t+1$, and if $\bar{a}[j] = 0$ we set it to $5 \cdot t+1$.
\item If $\bar{a}[j] = 1$ we set the weight of the edge $(g_j,a)$ to $2 \cdot t+1$, and if $\bar{a}[j] = 0$ we set it to $0 \cdot t+1$.
\end{itemize}

\item Similarly, every node $a' \in T$ is connected to and from every node $g_j$ in $J$ and the weights depend on $\bar{a}$. 
\begin{itemize}
\item If $\bar{a}[j] = 1$ we set the weight of the edge $(g_j,a')$ to $5 \cdot t+1$, and if $\bar{a'}[j] = 0$ we set it to $3 \cdot t+1$.
\item If $\bar{a}[j] = 1$ we set the weight of the edge $(a',g_j)$ to $0 \cdot t+1$, and if $\bar{a}[j] = 0$ we set it to $2 \cdot t+1$.
\end{itemize}

\item Finally, every node in $J$ is connected with bi-directional edges of weight $3 \cdot t +1$ to each of the o-nodes $\{o_1,o_2,o_3,o_4 \}$.

\end{itemize}

This completes the reduction. The new graph has $O(n)$ nodes and $O(n \log n)$ edges.

\paragraph*{Correctness}

The correctness of the reduction follows from the next two lemmas.

\begin{lemma}
\label{lem:YEScase}
If node $a \in V_1$ is not in a $k$-cycle in $G$ then $\rtd_{G'}(a,a')\geq 10\cdot t$ where $a \in S, a' \in T$ are the two copies of $a$ in $G'$.
\end{lemma}

\begin{lemma}
\label{lem:NOcase}
Suppose that all nodes $a \in V_1$ are in a $k$-cycle in $G$, then $\rtd_{G'}(x,y) \leq 6 \cdot t + 2k$ for all pairs $x,y \in V(G')$.
\end{lemma}

The two lemmas will become evident after we establish a series of claims about the distances in $G'$.

Let us begin with the interesting pairs $(a,a') \in S \times T$ where $a,a'$ are the two copies in $G'$ of a node $a \in G$.
The next claim shows that in the ``good'' case where $a$ is in a $k$-cycle, the roundtrip distance is $6t+O(k)$.

\begin{claim}
If node $a \in V_1$ is in a $k$-cycle in $G$ then $\rtd_{G'}(a,a') \leq 6t +2k$.
\end{claim}

\begin{proof}
This holds because of the forwards and backwards edges defined in the beginning. 
The edges of the $k$-cycle correspond to a forwards path from $a$ to $a'$ via the nodes in $V_2^{fwd},\ldots,V_k^{fwd}$ and a backwards path from $a'$ to $a$ via the nodes in $V_2^{bwd},\ldots,V_k^{bwd}$.
The weight of the forwards path is $6 \cdot t +k$ and the weight of the backwards path is $0 \cdot t + k$.
\end{proof}

Note that if the node $a$ is not in a $k$-cycle then neither the forwards nor backwards paths that were used in the previous proof exist in $G'$.

Next, we show that the distance between any pair $a \in S, b' \in T$ for a distinct pair of nodes $a,b\in V_1$ is $\leq 6t +O(1)$ due to the bit-gadget $J$.

\begin{claim}
For any pair of nodes $a,b \in V_1$ such that $a \neq b$ we have $\rtd_{G'}(a,b') \leq 6t +4$ where $a\in S$ and $b' \in T$.
\end{claim}

\begin{proof}
Let $j \in [d]$ be the coordinate such that $\bar{a}[j]=1$ but $\bar{b}[j]=0$. Such a coordinates is guaranteed to exist because $a \neq b$.
The path $a \to g_j \to b'$ has weight $3 \cdot t +1 + 3 \cdot t +1 = 6 \cdot t +2$.
The path $b' \to g_j \to a$ has weight $0 \cdot t +1 + 0 \cdot t +1 = 0 \cdot t +2$.
Thus, the roundtrip distance is at most $6\cdot t + 4$.
\end{proof}

Note that for the interesting pairs $(a,a')$ the above argument breaks, and the gadget $J$ does not provide a path of length $<12t$.

So far we have established that if all nodes $a\in V_1$ are in a $k$-cycle then all pairs in $S \times T$ have roundtrip distance $6t+O(k)$.
Let us now bound the distances within $S$ and within $T$, also using the bit-gadget $J$.

\begin{claim}
For any pair of nodes $a,b \in V_1$ such that $a \neq b$ we have $\rtd_{G'}(a,b) \leq 6t +4$ and $\rtd_{G'}(a',b') \leq 6t +4$  where $a,b\in S$ and $a',b' \in T$.
\end{claim}

\begin{proof}
Let $j \in [d]$ be the coordinate such that $\bar{a}[j] = \bar{b}[j]=1$.
The path $a \to g_j \to b$ has weight $3 \cdot t +1 + 0 \cdot t +1 = 3 \cdot t +2$.
And for the same reason, 
the path $b \to g_j \to a$ also has weight $3 \cdot t +1 + 0 \cdot t +1 = 3 \cdot t +2$.
Thus, the roundtrip distance between $a$ and $b$ is at most $6\cdot t + 4$.

For the pair $a',b' \in T$ we make a similar argument but consider the coordinate $i \in [d]$ in which both identifiers are $0$ rather than $1$; i.e. $\bar{a}[j] = \bar{b}[j]=0$.
The path $a' \to g_i \to b'$ has weight $0 \cdot t +1 + 3 \cdot t +1 = 3 \cdot t +2$.
And the path $b' \to g_i \to a'$ has weight $0 \cdot t +1 + 3 \cdot t +1 = 3 \cdot t +2$.
Thus, the roundtrip distance is at most $6\cdot t + 4$.

\end{proof}

After upper bounding all distances among pairs in $S \cup T$, it remains to analyze the other nodes in the construction; fortunately the o-nodes make it easy to see that all such distance are upper bounded by $6t+O(1)$.

\begin{claim}
The roundtrip distance between any pair of nodes $u,v \in V_2^{fwd} \cup \cdots \cup V_k^{fwd} \cup V_2^{bwd} \cup \cdots \cup V_k^{bwd} \cup  J \cup \{o_1,o_2,o_3,o_4\}$ is at most $6t+6$.
\end{claim}

\begin{proof}

The upper bound holds trivially for all pairs in $J \cup \{o_1,o_2,o_3,o_4\}$ because there are bidirectional edges of weight $3 \cdot t+1$ between any pair of them. 

Let $\{u,v\}$ by any pair that is not already covered by the previous argument. It must have an endpoint in $ V_2^{fwd} \cup \cdots \cup V_k^{fwd} \cup V_2^{bwd} \cup \cdots \cup V_k^{bwd}$, let it be $u$.
Observe that $u$ can reach any $v$ with distance $3 \cdot t+3$ because $u$ is at distance $1$ to some node in $\{o_1,o_2,o_3,o_4\}$ and $v$ is at distance $\leq 3 \cdot t + 2$ from any node in $\{o_1,o_2,o_3,o_4\}$.
Moreover, $v$ can reach any node in $\{o_1,o_2,o_3,o_4\}$ with weight $\leq 3 \cdot t + 2$ and there is some node in $\{o_1,o_2,o_3,o_4\}$ that can reach $u$ with weight $1$.
Thus, the roundtrip distance is at most $\leq 3 \cdot t + 6$.

\end{proof}

Finally, it remains to bound the distances for pairs with one endpoint in $S\cup T$ and one endpoint in the rest of $G'$. 
This will be broken into two claims, each using a different simple argument.

\begin{claim}
For any nodes $a \in S, a' \in T, g \in J$ we have $\rtd(a,g), \rtd(a,g') \leq 6\cdot t+2$.
\end{claim}

\begin{proof}
The direct roundtrips $a \to g \to a$ and $a' \to g \to a'$ have the desired distance.
\end{proof}

\begin{claim}
For any nodes $a \in S, a' \in T, v \in V_2^{fwd} \cup \cdots \cup V_k^{fwd} \cup V_2^{bwd} \cup \cdots  \cup V_k^{bwd} \cup \{ o_1,o_2,o_3,o_4 \}$ we have $\rtd(a,v), \rtd(a',v) \leq 6\cdot t+4$.
\end{claim}

\begin{proof}
\begin{itemize}
\item For $a \in S$ and any node $v \in V_2^{fwd} \cup \cdots \cup V_k^{fwd}$ the roundtrip $a \to o_1 \to v \to o_2 \to a$ has weight $6\cdot t+4$. Thus, $\rtd(a,v),\rtd(a,o_1),\rtd(a,o_2) \leq 6\cdot t+4$.
\item For $a \in S$ and any node $v \in V_2^{bwd} \cup \cdots \cup V_k^{bwd}$ the roundtrip $a \to o_3 \to v \to o_4 \to a$ has weight $6\cdot t+4$. Thus, $\rtd(a,v),\rtd(a,o_3),\rtd(a,o_4) \leq 6\cdot t+4$.
\item For $a' \in T$ and any node $v \in V_2^{fwd} \cup \cdots \cup V_k^{fwd}$ the roundtrip $a' \to o_1 \to v \to o_2 \to a'$ has weight $6\cdot t+4$. Thus, $\rtd(a',v),\rtd(a',o_1),\rtd(a',o_2) \leq 6\cdot t+4$.
\item For $a' \in T$ and any node $v \in V_2^{bwd} \cup \cdots \cup V_k^{bwd}$ the roundtrip $a' \to o_3 \to v \to o_4 \to a'$ has weight $6\cdot t+4$. Thus, $\rtd(a',v),\rtd(a',o_3),\rtd(a',o_4) \leq 6\cdot t+4$.

\end{itemize}

\end{proof}

The above claims suffice to establish Lemma~\ref{lem:NOcase} because we have upper bounded the roundtrip-diameter by $6t+2k$ in the case that all nodes in $G$ are in a $k$-cycle.

The next series of claims lower bound the roundtrip-distance between a pair $\{a,a'\}$ in the case that $a$ is not in a $k$-cycle in $G$.
In this case, there is simply no path from $a$ to $a'$ (or in the other direction) that avoids one of the o-nodes or the bit-gadget $J$. 
Therefore, our proof strategy is to lower bound the weight of any path that uses these nodes.
In these arguments we will ignore the $+1$ in the weights of edges and treat them as zero.

\begin{claim}
\label{claim:O}
Any path from $a\in S$ to $a'\in T$ that uses one of the nodes in $\{o_1,o_2,o_3,o_4\}$ must have distance at least $8t$.
\end{claim}

\begin{proof}
To establish the claim we lower bound the distances between the nodes in $S,T$ and the o-nodes.

\begin{itemize}
\item $\dist(a,o_2) \geq 3t$ because, in fact, there are no edges leaving $S$ that are cheaper than $3t$.
\item $\dist(o_1,a') \geq 3t$ because there are no edges entering $T$ that are cheaper than $3t$.
\item $\dist(a,o_1) \geq 5t$ because the direct edge has weight $5t+1$ and all other edges entering $o_1$ have weight $\geq 3t$ plus all edges leaving $a$ have weight $\geq 3t$, meaning that any path of length at least two will have weight $\geq 6t$.
\item $\dist(o_2,a') \geq 5t$ because the direct edge has weight $5t+1$ and all other edges leaving $o_2$ have weight $\geq 3t$ and all edges entering $a'$ have weight $\geq 3t$.
\end{itemize}
By the above bounds on the distances we can see that any path from $a$ to $a'$ that goes through $o_1$ or $o_2$ must have distance $\geq 8t$.
The following bounds address the paths that use $o_3$ or $o_4$.

\begin{itemize}
\item $\dist(a,o_3),\dist(a,o_4) \geq 4t$ because only nodes in $V_2^{fwd} \cup \cdots \cup V_k^{fwd} \cup \{ o_2 \}$ may be reachable from $S$ with distance $<4t$.
\item $\dist(o_3,a'),\dist(o_4,a') \geq 4t$ because only nodes in $V_2^{fwd} \cup \cdots \cup V_k^{fwd} \cup \{ o_2 \}$ may reach $T$ with distance $<4t$.
\end{itemize}

\end{proof}

\begin{claim}
\label{claim:J}
If $a \in V_1$ is not in a $k$-cycle in $G$ then any path from $a\in S$ to $a'\in T$ that does not use one of the nodes in $\{o_1,o_2,o_3,o_4\}$ must have distance at least $8t$.
\end{claim}

\begin{proof}
If the node $a$ is not in a $k$-cycle in $G$ then there are only two ways that at a path of distance $<8t$ from $a \in S$ to $a' \in T$ could possibly go, without using any of the o-nodes: either by first going to another node $b \in S$ and then going from $b$ to $a'$, or by first going to a node $b' \in T$ and then going from $b'$ to $a'$.
This is because any path via $V_2^{fwd} \cup \cdots V_k^{fwd}$ corresponds to a $k$-cycle in $G$, which is assumed to be inexistent, and any path of weight $<8t$ via the bit-gadget $J$ corresponds to a coordinate $j \in [d]$ in which the two identifiers differ which is also inexistent (since both are $\bar{a}$).
In either case, the path will have length $\geq 9t$ due to the following observations:
\begin{itemize}
\item For any pair $a,b \in S$ we have $\dist(a,b) \geq 3t$. This is because all edges leaving $S$ have weight $\geq 3t$.
\item For any pair $a',b' \in T$ we have $\dist(a',b') \geq 3t$. This is because all edges entering $T$ have weight $\geq 3t$.
\item For any pair $a\in S,b' \in T$ we have $\dist(a,b') \geq 6t$. This is because all edges leaving $S$ or entering $T$ have weight $\geq 3t$ and moreover there are no direct edges from $S$ to $T$.
\end{itemize}
\end{proof}

\begin{claim}
Any path from $a'\in T$ to $a\in S$ that uses one of the nodes in $\{o_1,o_2,o_3,o_4\}$ must have distance at least $2t$.
\end{claim}

\begin{proof}
The proof is analogous to that of Claim~\ref{claim:O}. Let us lower bound the distances between the nodes in $S,T$ and the o-nodes.
\begin{itemize}
\item $\dist(a',o_3) \geq 2t$ because there are no edges entering $o_3$ with weight less than $2t$.
\item $\dist(o_4,a) \geq 2t$ because there are no edges leaving $o_4$ with weight less than $2t$.
\end{itemize}
This implies that any path from $a'$ to $a$ that goes through $o_3$ or $o_4$ must have distance $\geq 2t$.
The following bounds address the paths that use $o_1$ or $o_2$.

\begin{itemize}
\item $\dist(a',o_1),\dist(a',o_2) \geq t$ because only nodes in $V_2^{bwd} \cup \cdots \cup V_k^{bwd} \cup \{ o_4 \}$ may be reachable from $T$ with distance $<t$.
\item $\dist(o_1,a),\dist(o_2,a) \geq t$ because only nodes in $V_2^{bwd} \cup \cdots \cup V_k^{bwd} \cup \{ o_3 \}$ may reach $T$ with distance $<t$.
\end{itemize}

\end{proof}

\begin{claim}
If $a \in V_1$ is not in a $k$-cycle in $G$ then any path from $a'\in T$ to $a\in S$ that does not use one of the nodes in $\{o_1,o_2,o_3,o_4\}$ must have distance at least $2t$.
\end{claim}

\begin{proof}
The proof is analogous to that of Claim~\ref{claim:J}.
A direct path via $V_2^{bwd} \cup \cdots V_k^{bwd}$ does not exist, and a direct path via the $J$ gadget has weight $\geq 2t$.
Thus, a path of weight $<2t$ from $a'$ to $a$ must either visit a node $b\in S$ or a node $b' \in T$. In either case the distance will be $\geq 3t$ by the following bounds:

\begin{itemize}
\item For any pair $a,b \in S$ we have $\dist(a,b) \geq 3t$ because all edges leaving $S$ have weight $\geq 3t$.
\item For any pair $a',b' \in T$ we have $\dist(a',b') \geq 3t$ because all edges entering $T$ have weight $\geq 3t$.
\end{itemize}
\end{proof}

As a result of the above four claims, we know that if $a$ is not in a $k$-cycle in $G$ then the roundtrip distance between $a \in S$ and $a' \in T$ is at least $8t + 2t$ which establishes Lemma~\ref{lem:YEScase}.
Together, Lemma~\ref{lem:NOcase} and~\ref{lem:YEScase} show the correctness of the reduction. An algorithm that can distinguish between roundtrip-diameter $\geq 10t$ from roundtrip-diameter $\leq 6t +2k$ can solve the All-Nodes $k$-Cycle problem.
By choosing $t$ to be a large enough constant, this can be achieved by an algorithm for roundtrip-diameter with approximation factor $5/3-\eps$.
\end{proof}



\bibliographystyle{alpha}
\bibliography{bib}

\appendix

\section{General approximation of directed (one-way) diameter}
\label{sec:alg-main}
We now give our general algorithm, generalizing the algorithm from Section~\ref{sec:alg-74} and proving Theorem~\ref{thm:alg-main}.
\begin{theorem*}[Theorem~\ref{thm:alg-main}, restated]
  Let $k=2^{t+2}$ for nonnegative integer $t\ge 0$.
  For every $\varepsilon>0$, there exists an $2-\frac{1}{k}+\varepsilon$ approximation of diameter in directed weighted graphs in time $\tilde O(m^{1+\alpha}/\varepsilon)$, for 
  \begin{align}
  \alpha=\frac{2(\frac{2}{\omega-1})^t - \frac{(\omega-1)^2}{2}}{(\frac{2}{\omega-1})^t(7-\omega) - \frac{\omega^2-1}{2}}.
  \label{eq:alpha-formula}
  \end{align}
\end{theorem*}
Note that for $t=0$, this recovers Theorem~\ref{thm:alg-74}, with a lost $\varepsilon$ factor in the approximation.
\begin{proof}
  Similar to Theorem~\ref{thm:alg-74}, it suffices to show that, for any positive integer $D > 0$, there exists an algorithm $\mathcal{A}_D$ running in time $\tilde O(m^{1+\alpha})$ that takes as input any graph and accepts if the diameter is at least $D$, rejects if the diameter is less than $(\frac{k}{2k-1}-\varepsilon)D$, and returns arbitrarily otherwise.
  Then with a binary search argument we can get a $2-\frac{1}{k}+4\varepsilon$-approximation for every small $\epsilon>0$. Replacing $\varepsilon$ with $\varepsilon/4$ gives the result.

One can check that our choice of $\alpha$ guarantees a unique sequence of numbers $1-\alpha = \alpha_0 > \alpha_1 > \cdots > \alpha_t > \alpha_{t+1} = \alpha$ such that
\begin{align}
  2\alpha = \alpha_i + (1-\alpha_{i+1})\frac{\omega-1}{2}
  \label{eq:alg-0}
\end{align}
for $i=0,\dots,t$. We can determine $\alpha$ by using $\alpha_0=1-\alpha$ and $\alpha_{t+1}=\alpha$ and iterating the recursion \eqref{eq:alg-0} to obtain an equation for $\alpha$, which we solve to get \eqref{eq:alpha-formula}.\footnote{
Here are some details. First rewrite \eqref{eq:alg-0} as
\begin{align}
  \alpha_{i+1} = \frac{2}{\omega-1}\alpha_i + 1 - \frac{4}{\omega-1}\alpha
\end{align}
For $i=0$, since $\alpha_0=1-\alpha$ we have $\alpha_1=\frac{2}{\omega-1}+1- \frac{6}{\omega-1}\alpha$ by \eqref{eq:alg-0}. If we define $\alpha_i=\beta_i-\gamma_i\alpha$, we have that $\beta_1=\frac{\omega+1}{\omega-1}$ and $\gamma_1=\frac{6}{\omega-1}$. Using equation \ref{eq:alg-0} for $i\le t-1$, we have $\beta_{i+1}=\frac{2}{\omega-1}\beta_i+1$ and $\gamma_{i+1}=\frac{4}{\omega-1}+\frac{2}{\omega-1}\gamma_i$. So we have $\beta_{i+1}=(\frac{2}{\omega-1})^i(\beta_1-\frac{\omega-1}{\omega-3})+\frac{\omega-1}{\omega-3}$, and $\gamma_{i+1}=(\frac{2}{\omega-1})^i(\gamma_1-\frac{4}{\omega-3})+\frac{4}{\omega-3}$. So $\beta_t=(\frac{2}{\omega-1})^t\frac{2}{3-\omega}+\frac{\omega-1}{\omega-3}$ and $\gamma_t=(\frac{2}{\omega-1})^t\frac{7-\omega}{3-\omega}+\frac{4}{\omega-3}$.
Using equation \eqref{eq:alg-0} for $i=t$ we have $\alpha = \frac{2}{\omega-1}\alpha_t+1-\frac{4}{\omega-1}\alpha$, so we have 
\begin{equation}
    \alpha=\frac{\beta_t+\frac{\omega-1}{2}}{\gamma_t+\frac{\omega+3}{2}}=\frac{2(\frac{2}{\omega-1})^t-\frac{(\omega-1)^2}{2}}{(\frac{2}{\omega-1})^t(7-\omega)-\frac{\omega^2-1}{2}}.
\end{equation}
}
For such an $\alpha$, we can check $\alpha_0 > \alpha_1 > \cdots > \alpha_t > \alpha_{t+1}$.\footnote{
Here are some details: First check that $\alpha\ge \frac{2}{7-\omega}$ from \eqref{eq:alpha-formula} and $2\le\omega\le 3$. 
Combining this with \eqref{eq:alg-0} at $i=0$ gives $\alpha_0 > \alpha_1$.
Additionally, subtracting the $i$ and $i+1$ versions of equation \eqref{eq:alg-0}, we see that $\alpha_i > \alpha_{i+1}$ implies $\alpha_{i+1}> \alpha_{i+2}$, so by induction we indeed have $\alpha_0 > \alpha_1 > \cdots > \alpha_t > \alpha_{t+1}$.
}

\paragraph*{The algorithm.} We now describe the algorithm $\mathcal{A}_D$. 
\begin{enumerate}
  \item First, we apply a standard trick that replaces the input graph on $n$ vertices and $m$ edges with an $2m$-vertex graph of max-degree-3 that preserves the diameter: replace each vertex $v$ with a cycle of degree$(v)$ new vertices with weight-0 edges and where the edges to $v$ now connect to distinct vertices of the cycle. From now on, we work with this max-degree-3 graph on $2m$ vertices. 
  Now running Dijkstra's algorithm until $m^\beta$ vertices are visited takes $\tilde O(m^\beta)$ time, since it costs $\tilde O(1)$ time to visit each vertex and it's at-most-$3$ edges in Dijkstra's algorithm.

  \item\label{step:1} Sample $4m^{\alpha}\log m$ uniformly random vertices and compute their eccentricities. If any such vertex has (in- or out-) eccentricity at least $\frac{k}{2k-1}D$ Accept.

  \item\label{step:2} For every vertex $v$, determine if $|B_{\frac{1}{2k-1}D}^{out}(v)|\le m^{\alpha}$. If such a vertex $v$ exists, determine if any vertex of $B_{\frac{1}{2k-1}D}^{out}(v)$ has eccentricity at least $\frac{k}{2k-1}D$, and Accept if so.

  \item\label{step:3} For every vertex $v$, determine if $|B_{\frac{1}{2k-1}D}^{in}(v)|\le m^{\alpha}$. If such a vertex $v$ exists, determine if any vertex of $B_{\frac{1}{2k-1}D}^{in}(v)$ has eccentricity at least $\frac{k}{2k-1}D$, and Accept if so.

  \item\label{step:4} For $i=0,\dots,t$:

  \begin{enumerate}
    \item\label{step:4a} Sample $4m^{1-\alpha_{i+1}}\log m$ uniformly random vertices $\hat S$. 
    For each vertex in $\hat S$, run partial in- and out-Dijkstra each until $m^{\alpha_i}$ vertices have been visited.
    Compute
    \begin{align}
      S^{out} &= \left\{s\in\hat S: \abs{B^{out}_{\frac{2^{i+1}}{2k-1}D}(s)}\le m^{\alpha_i}\right\} \nonumber\\
      S^{in} &= \left\{s\in\hat S: \abs{B^{in}_{\frac{2^{i+1}}{2k-1}D}(s)}\le m^{\alpha_i}\right\}
    \end{align}
    and record the distances $\hat d(s,v)$ from the partial out-Dijkstra for $s\in S^{out}$, for $v\in B^{out+}_{\frac{2^{i+1}}{2k-1}D}(s)$.
    Note that $\hat d(s,v) \ge d(s,v)$ for all such $v$ with equality if $v\in B^{out}_{\frac{2^{i+1}}{2k-1}D}(s)$.
    Similarly, record the distances $\hat d(v,s)$ from the partial in-Dijkstra for $s\in S^{in}$, for $v\in B^{in+}_{\frac{2^{i+1}}{2k-1}D}(s)$.

    \item\label{step:4b} Sample $4m^{1-\alpha}\log m$ uniformly random vertices $\hat T$.
    For each vertex in $\hat T$, run partial in- and out- Dijkstra until $m^{\alpha_i}$ vertices have been visited.
    Compute
    \begin{align}
      T^{out}&=\left\{t\in\hat T: \abs{B^{out}_{\frac{k-2^{i+1}}{2k-1}D}(t)}\le m^{\alpha_i}\right\} \nonumber\\
      T^{in}&=\left\{t\in\hat T: \abs{B^{in}_{\frac{k-2^{i+1}}{2k-1}D}(t)}\le m^{\alpha_i}\right\} 
    \end{align}
    and record the distances $\hat d(t,v)$ from the partial out-Dijkstra for $t\in T^{out}$, for $v\in B^{out+}_{\frac{k-2^{i+1}}{2k-1}D}(s)$, and similarly, record the distances $\hat d(v,t)$ from the partial in-Dijkstra for $t\in T^{in}$, for $v\in B^{in+}_{\frac{k-2^{i+1}}{2k-1}D}(s)$.\footnote{Note that if $a'\in S^{out}\cup T^{out}$ and $b'\in S^{in}\cup T^{in}$, $\hat d(a',b')$ may be recorded multiple times, with different values. We take the smallest one, as this only helps us.}

    \item\label{step:4c} For integers $0\le j\le \frac{1}{\varepsilon}\cdot \frac{k}{2k-1}$, construct the following matrices
    \begin{itemize}
    \item $A^{j,out}\in \mathbb{R}^{S^{out}\times V}$ where $A_{s,v}^{j,out} = 1$ if $\hat d(s,v)\le j\varepsilon D$, and all other entries are zero. 
    \item $A^{j,in}\in \mathbb{R}^{V\times S^{in}}$  where $A_{v,s}^{j,in}=1$ if $\hat d(v,s)\le (\frac{k}{2k-1}-j\varepsilon)D$ and all other entries are zero. 
    \item $B^{j,out}\in \mathbb{R}^{T^{out}\times V}$ where $B_{t,v}^{j,out} = 1$ if $\hat d(t,v)\le j\varepsilon D$, and all other entries are zero. 
    \item $B^{j,in}\in \mathbb{R}^{V\times T^{in}}$  where $B_{v,t}^{j,in}=1$ if $\hat d(v,t)\le (\frac{k}{2k-1}-j\varepsilon)D$ and all other entries are zero. 
    \end{itemize}
    For all $j$, compute $A^{j,out}\cdot B^{j,in}\in\mathbb{R}^{S^{out}\times T^{in}}$ and $B^{j,out}\cdot A^{j,in}\in\mathbb{R}^{T^{out}\times S^{in}}$ using sparse matrix multiplication. 
    If there exists $s\in S^{out}$ and $t\in T^{in}$ such that $(A^{j,out}\cdot B^{j,in})_{s,t} = 0$ for all $j$, Accept.
    If there exists $t\in T^{out}$ and $s\in S^{in}$ such that $(B^{j,out}\cdot A^{j,in})_{t,s} = 0$ for all $j$, Accept.
    Otherwise Reject.
  \end{enumerate}
\end{enumerate}

\paragraph*{Runtime.} 
  Similar to Theorem~\ref{thm:alg-74}, Steps~\ref{step:1}, \ref{step:2}, and \ref{step:3} take time $\tilde O(m^{1+\alpha})$. 
  For Step~\ref{step:4a}, like in Theorem~\ref{thm:alg-74}, we can compute $S_i^{out}$ and $S_i^{in}$ and determine the desired distances in time  $\tilde O(m^{1-\alpha_{i+1}+\alpha_i})$ using partial Dijkstra. 
  Similarly in Step~\ref{step:4b}, we can compute $T_i^{out}$ and $T_i^{in}$ and the desired distances in time $\tilde O(m^{1-\alpha+\alpha_i})$.
  In Step~\ref{step:4c}, the runtime is the time to multiply sparse matrices. Each matrix $A^j$ has $\tilde O(m^{1-\alpha_{i+1}})$ rows, $m$ columns, and sparsity $\tilde O(m^{1-\alpha_{i+1}+\alpha_i})$, and matrix $B^j$ has $m$ rows, $\tilde O(m^{1-\alpha})$ columns, and sparsity $\tilde O(m^{1-\alpha+\alpha_i})$. 
  We can compute the product $A^{j,out}\cdot B^{j,in}$ by breaking into $m^{\alpha_{i+1}-\alpha}$ matrix multiplications of dimension $(\tilde O(m^{1-\alpha_{i+1}}), n, \tilde O(m^{1-\alpha_{i+1}}))$, where each matrix has sparsity $\tilde O(m^{1-\alpha_{i+1}+\alpha_i})$ (because each row of $A^{j,out}$ and each column of $B^{j,in}$ has sparsity $O(m^{\alpha_i})$).
  Each submatrix multiplication runs in time $\tilde O(m^{1-\alpha_{i+1}+\alpha_i}m^{(1-\alpha_{i+1})\frac{\omega-1}{2}})$ by Lemma~\ref{lem:matmul}.
  To apply Lemma~\ref{lem:matmul}, we need $1-\alpha_{i+1}+\alpha_i \ge (1-\alpha_{i+1})\frac{\omega+1}{2}$, which holds by rearranging \eqref{eq:alg-0} and using $\alpha_i > \alpha$.
  Thus, one product $A^{j,out}\cdot B^{j,in}$ takes time $\tilde O(m^{1-\alpha+\alpha_i}m^{(1-\alpha_{i+1})\frac{\omega-1}{2}})$, and so, as there are $O(1/\varepsilon)$ matrix multiplications, Step~\ref{step:4c} takes time $\tilde O(m^{1-\alpha+\alpha_i}m^{(1-\alpha_{i+1})\frac{\omega-1}{2}}/\varepsilon)$.
  Thus, the total runtime is
  \begin{align}
    \tilde O(m^{1+\alpha}) + \sum_{i=0}^{t} \tilde O(m^{1-\alpha+\alpha_i + (1-\alpha_{i+1})\frac{\omega-1}{2}})
    \le \tilde O(m^{1+\alpha})
  \end{align}
  as desired, where the bound follows from \eqref{eq:alg-0}.

  \paragraph*{If the diameter is less than $(\frac{k}{2k-1}-\varepsilon)D$, we always reject.}
  Clearly every vertex has eccentricity less than $\frac{k}{2k-1}$, so we indeed do not accept at Steps~\ref{step:1}, \ref{step:2}, and \ref{step:3}.
  At Step~\ref{step:4c}, consider any $s\in S^{out}$ and $t\in T^{in}$.
  By definition, we have $d(s,t) < (\frac{k}{2k-1}-\varepsilon)D$.
  Let $v$ be the latest vertex on the $s$-to-$t$ shortest path such that $d(s,v) \le \frac{2^{i+1}}{2k-1}D$ and let $v'$ be the following vertex, if it exists.
  Then $v$ is in $B^{out}_{\frac{2^{i+1}}{2k-1}D}(s)$ and $v'$, if it exists, is in $B^{in}_{\frac{k-2^{i+1}}{2k-1}D}(t)$.
  Thus, $v$ is visited in the partial Dijkstra from $s$, so $\hat d(s,v)=d(s,v)$ is accurate.
  Similarly, either $v=t$ so that $\hat d(v,t)$ is accurately 0, or $v'$ exists and is visited in the partial Dijkstra from $t$, so that $\hat d(v,t)$ is updated to be at most $w_{v,v'} + \hat d(v',t) = w_{v,v'}+d(v',t) = d(v,t)$, and thus is accurate. We used in the second equality that the $v$-to-$t$ shortest path goes through $v'$. 
  We thus have $\hat d(s,v) + \hat d(v,t) < (\frac{k}{2k-1}-\varepsilon)D$.
  Setting $j = \ceil{\hat d(s,v)/\varepsilon}$, we have $A^{j,out}_{s,v} = 1$ and $\hat d(v,t) = (\frac{k}{2k-1}-\varepsilon)D - \hat d(s,v) \le (\frac{k}{2k-1}-j\varepsilon)D$, so $B^{j,in}_{v,t} = 1$.
  Hence, $(A^{j,out}\cdot B^{j,in})_{s,t}\ge 1$, and this holds for any $s\in S^{out}$ and $t\in T^{in}$.
  Similarly, $(A^{j,in}\cdot B^{j,out})_{t,s}\ge 1$ for any $t\in T^{out}$ and $s\in S^{in}$.
  Thus, we do not accept at Step~\ref{step:4c}, so we reject, as desired.

  \paragraph*{If the diameter is at least $D$, we accept with high probability.}
  Let $a$ and $b$ be vertices at distance $d(a,b)\ge D$.
  Similar to Theorem~\ref{thm:alg-74}, we may assume all the folloiwng hold, or else we accept with high probability at one of Steps~\ref{step:1}, \ref{step:2}, or \ref{step:3}.
  \begin{align}
    |B_{\frac{k-1}{2k-1}D}^{out}(a)|\le m^{1-\alpha}, \quad
    |B_{\frac{k-1}{2k-1}D}^{in}(b)|\le m^{1-\alpha}, \quad
    |B_{\frac{1}{2k-1}D}^{out}(a)| > m^{\alpha}, \quad
    |B_{\frac{1}{2k-1}D}^{in}(b)| > m^{\alpha}
    \label{eq:alg-1}
  \end{align}
  Let $i\in\{0,\dots,t+1\}$ be the largest index such that $|B_{\frac{k-2^{i+1}+1}{2k-1}D}^{out}(a)|\le m^{\alpha_i}$, $|B_{\frac{k-2^{i+1}+1}{2k-1}D}^{in}(b)|\le m^{\alpha_i}$.
  By the first half of \eqref{eq:alg-1}, $i$ exists, and by the second of half of \eqref{eq:alg-1}, $i < t+1$. Thus, we have
  \begin{align}
    |B_{\frac{k-2^{i+1}+1}{2k-1}D}^{out}(a)|\le m^{\alpha_i},\text{ and } 
    |B_{\frac{k-2^{i+1}+1}{2k-1}D}^{in}(b)|\le m^{\alpha_i},\text{ and }\nonumber\\
    \text{ either }|B_{\frac{k-2^{i+2}+1}{2k-1}D}^{out}(a)|> m^{\alpha_{i+1}} \text{ or }
    |B_{\frac{k-2^{i+2}+1}{2k-1}D}^{in}(b)|> m^{\alpha_{i+1}}
  \end{align}

  We now prove that iteration $i$ of Step~\ref{step:4} accepts.
  Suppose that $|B_{\frac{k-2^{i+2}+1}{2k-1}D}^{out}(a)|> m^{1-\alpha_{i+1}}$. The case $|B_{\frac{k-2^{i+2}+1}{2k-1}D}^{in}(b)|> m^{1-\alpha_{i+1}}$ is similar.
  With high probability $\hat S$ has a vertex $s$ in $B_{\frac{k-2^{i+2}+1}{2k-1}D}^{out}(a)$ by Lemma~\ref{lem:hit}.
  As $|B_{\frac{2^{i+1}}{2k-1}D}^{out}(s)|\le |B_{\frac{k-2^{i+1}+1}{2k-1}D}^{out}(a)|\le m^{\alpha_i}$, we have $s\in S^{out}$.
  With high probability $\hat T$ has a vertex $t$ in $B_{\frac{1}{2k-1}D}^{in}(b)$ by the bound in \eqref{eq:alg-1} and Lemma~\ref{lem:hit}.
  As $|B_{\frac{k-2^{i+1}}{2k-1}D}^{in}(t)|\le |B_{\frac{k-2^{i+1}+1}{2k-1}D}^{in}(b)|\le m^{\alpha_i}$, we have $t\in T^{in}$.
  By choice of $s$ and $t$, the triangle inequality gives
  \begin{align}
  d(s,t)\ge d(a,b) - d(a,s) - d(t,b) \ge D - \frac{k-2^{i+2}+1}{2k-1}D - \frac{1}{2k-1}D \ge \frac{k+1}{2k-1}.
  \label{eq:alg-3}
  \end{align}
  Note that if $(A^{j,out}\cdot B^{j,in})_{s,t} > 0$ for some $j$, then there exists a vertex $v$ such that $d(s,v) \le j\varepsilon D$ and $d(v,t)\le (\frac{k}{2k-1}-j\varepsilon)D$, so $d(s,t) \le d(s,v)+d(v,t) \le \frac{k}{2k-1}D$, contradicting \eqref{eq:alg-3}.
  Thus, $(A^{j,out}\cdot B^{j,in})_{s,t} = 0$, so we accept, as desired.
\end{proof}

\begin{remark}
  If the edge weights are integers $\{0,\dots, C\}$, we can get remove the $+\varepsilon$ and get a $\frac{2k-1}{k}$-approximation in time $\tilde O(m^{1+\alpha}C)$.
  We can set $\varepsilon=1/D$, and in Step~\ref{step:4c}, we only need matrix multiplications for $j\in[\frac{2^{i+1}}{2k-1}D, \frac{2^{i+1}}{2k-1}D+C]$, since crossing an edge changes distance by at most $C$, saving the $1/\varepsilon$ factor from the number of matrix multiplications.
  Furthermore, because the diameter is an integer, we can stop the binary search after $\log D\le \log(nC)$ steps.
\end{remark}

\section{Unweighted Roundtrip $2-\varepsilon$ hardness from $\ell_{\infty}$-CP}
\label{sec:linfty-unw}
We now prove Theorem~\ref{thm:lb-linfty-unw}, extending the proof from Section~\ref{sec:linfty} to unweighted graphs.
\begin{theorem}[Theorem~\ref{thm:lb-linfty-unw}, restated]
  Let $\alpha\ge 2, \gamma\in(0,1), \delta>0$. 
  If there is an $2-\frac{1}{\alpha}-\gamma$ approximation algorithm in time $O(m^{2-\delta})$ for roundtrip diameter in unweighted graphs, there is a $\alpha$-approximation for $\ell_\infty$-Closest-Pair on vectors of dimension $d\le n^{1-\delta}$ in time $\tilde O(n^{2-\delta})$.
\end{theorem}
\begin{proof}
  Let $M \ge 20/\gamma$ be a constant, $\beta \defeq \floor{M\alpha - 1}$, $\varepsilon = \frac{1}{4(\beta+1.5)}$.
  For convenience, we assume that $M$ is such that the fractional part of $M\alpha$ is less than 0.5, so that $\beta> M\alpha-1.5$.

  By Lemma~\ref{lem:bounded}, it suffices to find an algorithm for a $(0.5+\epsilon)\alpha$-bounded instance $I'=\{v_1',\ldots,v_n'\}$ of the $\ell_\infty$-Closest-Pair problem on vectors of dimension $d\varepsilon^{-1}\log n$ in time $\tilde O(n^{2-\delta})$.
  This algorithm needs to distinguish between the ``YES case,'' where there exists $i\neq j$ with $\|v_i'-v_j'\|_\infty\le 1$, and the ``NO case'' where $\|v_i'-v_j'\|_\infty\ge \alpha$ for all $i\neq j$.

  First we construct a new set of vectors $I=\{v_1,\ldots,v_n\}$, where for each $j\in[n]$ and $x\in [d]$,  $v_j[x]=\floor{M\cdot v_j[x]}$. This set of vectors has the following properties
  \begin{itemize}
  \item In the YES case, there exists $i\neq j\in [n]$ with $|v_i'[x]-v_j'[x]|\ge \alpha$ and thus $|v_i[x]-v_j[x]|\ge M\alpha-1\ge \beta$. 
  \item In the NO case, for all $i\neq j\in [n]$, we have $|v_i'[x]-v_j'[x]|\le 1$, then $|v_i[x]-v_j[x]|\le M$. 
  \item All entries of vectors in $I$ have absolute value at most $0.5\beta+2$: note that $v_j[x]=\floor{Mv_j'[x]}\le Mv_j'[x]\le (0.5+\epsilon)M\alpha < 0.5\beta + 2$ and $v_j[x]=\floor{Mv_j'[x]} > Mv_j'[x]-1\ge -(0.5+\epsilon)M\alpha-1 \ge -0.5\beta-2$. 
  \end{itemize}
  We now construct a graph $G$ in $O_{\gamma,\alpha}(nd)$ time such that, if there exists $|v'_i[x]-v'_j[x]|\ge \beta$, then the roundtrip diameter is at least $4\beta-2M$, and otherwise the diameter is a most $2\beta+8$.
  Indeed, this implies that a better than $\frac{4\beta-2M}{2\beta+8} = 2 - \frac{M+8}{\beta+4} \ge 2-\frac{1}{\alpha}-\gamma$ approximation of the roundtrip diameter can distinguish between the YES and NO case, solving the bounded $\ell_\infty$ instance, as desired. We now describe the graph.

\paragraph*{The graph.}

  \begin{figure}
    \begin{center}
    \includegraphics[width=\textwidth]{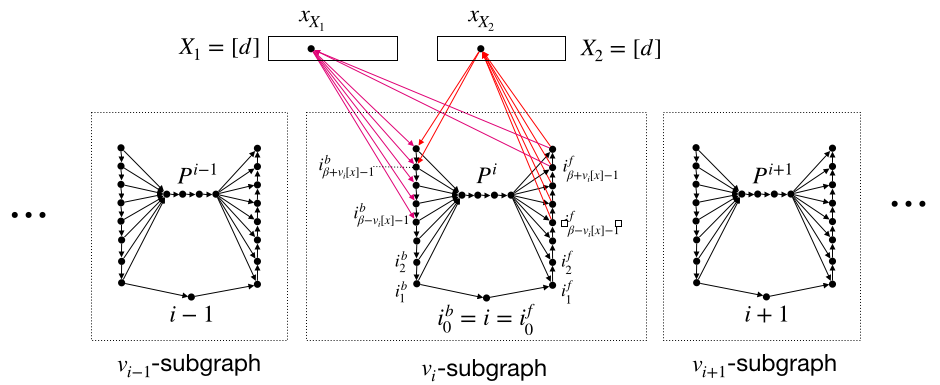}
    \end{center}
    \caption{The roundtrip diameter instance $G$. We modify the weighted lower bound graph, replacing each vector vertex $v_i$ with the \emph{$v_i$-subgraph}, illustrated above.}
    \label{fig:ellinfty_unw}
  \end{figure}

  The graph is illustrated in Figure~\ref{fig:ellinfty_unw}.
  For each $i\in [n]$ we first describe a subgraph called the \emph{$v_i$-subgraph}, which consists of the following. Let $x^*=\argmax_{x\in[d]}|v_i[x]|$.
  \begin{itemize}
  \item Vertices $i^{f}_1,\ldots,i^{f}_{\beta+|v_i[x^*]|-1}$ and $i^{b}_1,\ldots,i^{b}_{\beta+|v_i[x^*]|-1}$. Vertex $i=i^{f}_0=i^b_0$.  (superscripts $f$ and $b$ are for ``forward'' and ``backward'')
  \item For $j=0,\ldots,\beta+|v_i[x^*]|-2$, edges $(i_j^f,i_{j+1}^f)$ and $(i_{j+1}^b,i_j^b)$, so that the $i_j^f$ nodes construct a path of length $\beta+|v_i[x^*]|-1$ from $i$ and the $i_j^b$ nodes construct a path of length $\beta+|v_i[x^*]|-1$ to $i$.
  \item Vertices $P^i=\{p^i_1,\ldots,p^i_{\beta+1}\}$ which form a directed path $P^i$ of length $\beta$. Edges $(i_j^{b},p^i_1)$ and $(p^i_{\beta+1},i_j^f)$ for all $j=1,\ldots,\beta+|v_i[x^*]|-1$. 
  \end{itemize}
  The subgraph on the union of $\cup_ji^f_j$, $\cup_j i^b_j$ and $P^i$ is called the \emph{$v_i$-subgraph}, as all the nodes are associated to $v_i$.
  In addition to the $v_i$-subgraphs for all $i$, our graph has the following:
  \begin{itemize}
  \item As in the weighted case, we have two vertex sets $X_1=[d]$ and $X_2=[d]$ each identified by the coordinates. 
  \item For each $x\in [d]$, connect $i_{j}^f$ to $x_{X_1}$ and $x_{X_2}$ to $i_{j}^b$ for all $j\ge \beta+v_i[x]-1$. Connect $x_{X_1}$ to $i_{j}^b$ and $i_{j}^f$ to $x_{X_2}$ for all $j\ge \beta-v_i[x]-1$. These simulate the weighted edges from $S$ to $X_1$ and $X_2$ in the weighted construction. 
  \end{itemize}
  This finishes the construction.
  We now show that the roundtrip diameter is at most $2\beta+8$ in the NO case and at least $4\beta-M$ in the YES case.

\paragraph*{NO case.}
  We show the roundtrip distance between every pair of vertices $a$ and $c$ is at most $2\beta+8$. We break into the following cases.
  \begin{itemize}
  \item \textbf{Case 1: $a,c$ are both in the $v_i$-subgraph for some $i$.}
  Let $x^*=\argmax_{x\in[d]}(|v_i[x]|)$. Consider the two following cycles of length $2\beta$:
  \begin{align}
  i=i_0^f,\ldots,i_{\beta+v_i[x^*]-1}^f,x_{X_1}^*, i_{\beta-v_i[x^*]-1}^b, \ldots, i_0^b=i
  \end{align}
  \begin{align}
  i=i_0^f,\ldots,i_{\beta-v_i[x^*]-1}^f,x_{X_2}^*, i_{\beta+v_i[x^*]-1}^b, \ldots, i_0^b=i
  \end{align}
  These two cycles don't cover the following cases: (case 1) $a=i_j^f$ and $c=i_{j'}^b$ for $j,j' > \beta-|v_i[x^*]|-1$, (case 2) $a$ is in $P^i$. 
  Without loss of generality suppose $v_i[x^*]>0$. For the case 1, consider the following cycle 
  \begin{align}
  a,P^i, c=i_{j'}^b,\ldots, i_{\beta+v_i[x^*]-1},x_{X_1}^*,a
  \end{align}
  This cycle has length at most $\beta+2v_i[x^*]+2$. Since $v_i[x^*]\le 0.5\beta+2$, the cycle has length at most $2\beta+6$.
  Note that this also covers case 2 when $c\in\{i_j^f,i_j^f\}\cup P^i$ for some $j> \beta-|v_i[x^*]|-1$.

  For case 2, if $c\in \{i_j^f,i_j^f\}$ for $j\le \beta-|v_i[x^*]|-1$, consider the following cycle:

  \begin{align}
  i=i^f_0,\ldots,i^f_{\beta-v_i[x^*]-1}, x_{X_2}^*, i^b_{\beta+v_i[x^*]-1},P^i,i^f_{\beta+v_i[x^*]-1},x_{X_1}^*,i^b_{\beta-v_i[x^*]-1},\ldots,i_0^b=i
  \end{align}
  This cycle is of length at most $3\beta+1-2v_i[x^*]$. Note that $v_i[x^*]\ge 0.5\beta-2$. This is because if we consider some $j\in S$, then there is $y$ such that $|v_i[y]-v_j[y]|\ge \beta$. Since $|v_i[y]|,|v_j[y]|\le 0.5\beta+2$, we have that $|v_i[y]|,|v_j[y]|\ge 0.5\beta-2$. So $v_i[x^*]\ge |v_i[y]|\ge 0.5\beta-2$, and hence the length of the cycle is at most $2\beta+7$.

  \item \textbf{Case 2: $a$ is in the $v_i$-subgraph and $c$ is in the $v_j$-subgraph for $i\neq j$.}
  Let $x\in [d]$ be a coordinate where $|v_i[x]-v_j[x]|\ge \beta$. 
  Without loss of generality suppose that $v_i[x]>0>v_j[x]$. Note that $0.5\beta-2\le |v_i[x]|,|v_j[x]|\le 0.5\beta+2$.

  We show that there is a path from $x_{X_1}$ to $x_{X_2}$ of length at most $\beta+4$ that contains $a$. Similarly, we show that there is a path from $x_{X_2}$ to $x_{X_1}$ of length at most $\beta+4$ that contains $c$. Then the union of these two paths constructs a cycle of length at most $2\beta+8$ passing through $a$ and $c$. 

  If $a\in\{i_{k}^b,i_k^f\}$ for some $k\le \beta-v_i[x]-1$, consider this path of length $2\beta-2v_i[x]\le \beta+4$.
  \begin{align}
  x_{X_1},i_{\beta-v_i[x]-1}^b,\ldots,i_0^b=i_0^f,\ldots, i_{\beta-v_i[x]-1}^f,x_{X_2}
  \end{align}

  If $a\in\{i_k^b,i_k^f\}$ for $k>\beta-v_i[x]-1$, consider the following path of length $\beta+4$.
  \begin{align}
  x_{X_1},i_k^b,P^i,i_k^f, x_{X_2}
  \end{align}
  If $a\in P^i$, we consider the above path of length $\beta+4$.

  Now for $c$, we do a similar case analysis. If $c\in \{j_k^f,j_k^b\}$, for some $k\le \beta+v_j[x]-1$, consider the following path of length $2\beta+2v_j[x]\le \beta+4$. 

  \begin{align}
  x_{X_2},j_{\beta+v_j[x]-1}^b,\ldots,j_0^b=j_0^f,\ldots, j_{\beta+v_j[x]-1}^f,x_{X_1}
  \end{align}
  If $c\in\{j_k^b,j_k^f\}$ for $k>\beta+v_j[x]-1$, consider the following path of length $\beta+4$.
  \begin{align}
  x_{X_2},j_k^b,P^j,j_k^f, x_{X_1}
  \end{align}
  If $c\in P^j$, we consider the above path of length $\beta+4$.

  So the cycle is of length at most $2\beta+8$.

  \item \textbf{Case 3: $a$ is in the $v_i$-subgraph and $c\in X_1\cup X_2$.}
  Suppose $c=x_{X_1}$. If $a=i_{k}^b$ for $k\le \beta-v_i[x]-1$ or $a=i_{k'}^f$ for $k'\le \beta+v_i[x]-1$, then consider the following cycle of length $2\beta$.
  \begin{align}
  i=i_0^f,\ldots,i_{\beta+v_i[x]-1}^f,x_{X_1},i_{\beta-v_i[x]-1}^b,\ldots,i_0^b=i.
  \end{align}

  If $a=i_{k}^b$ for $k> \beta-v_i[x]-1$, consider the following cycle of length $\beta+4$:
  \begin{align}
  x_{X_1}, a,P^i,i_{\beta+v_i[x]-1}^f,x_{X_1}
  \end{align}
  If $a=i_{k}^f$ for $k> \beta+v_i[x]-1$, consider the following cycle of length $\beta+4$:
  \begin{align}
  x_{X_1}, i_{\beta-v_i[x]-1}^b,P^i,a,x_{X_1}
  \end{align}
  For $c=x_{X_2}$, everything is symmetric.

  \item \textbf{Case 4: $a,c\in X_1\cup X_2$.} Any two vertices in $X_1\cup X_2$ are at distance $\beta+4$ and thus roundtrip distance $2\beta+8$: for any $x,x'\in X_1\cup X_2$, pick any $i$. Then 
  \begin{align}
    x,i^b_{\beta+\|v_i\|_\infty-1},P^i,i^f_{\beta+\|v_i\|_\infty-1},x' 
  \end{align}
  is a path of length $\beta+4$.
  \end{itemize}
  This covers all cases, so we have shown that the roundtrip diameter in the NO case is at most $2\beta+8$.

\paragraph*{YES case.}
Suppose that there exist $i,j$ such that for all $x\in [d]$, $|v_i[x]-v_j[x]|\le M$. We show that $d(i,j)\ge 2\beta-M$. By symmetry, it follows that $d(j,i)\ge 2\beta-M$, so the roundtrip distance is at least $4\beta-2M$.

For every vertex $k$, we can check that $d(k,X_1\cup X_2), d(X_1\cup X_2,k) \ge \beta-\|v_k\|_\infty \ge 0.5\beta-2$.
If a path from $i$ to $j$ passes through a path $P^k$, then it must hit $X_1\cup X_2$ before and after path $P^k$ (even if $k=i$ or $k=j$), creating a path of length at least $\beta+4$ between two vertices of $X_1\cup X_2$.
Then the $i$-to-$j$ path has length at least 
\begin{align}
  d(i,X_1\cup X_2)+(\beta+4)+d(X_1\cup X_2,j) \ge (0.5\beta-2)+(\beta+4) + (0.5\beta-2) = 2\beta,
\end{align}
as desired.
If a path from $i$ to $j$ has a vertex $k\in[n]$, then the path must have length at least
\begin{align}
d(i,X_1\cup X_2)+d(X_1\cup X_2,k)+d(k,X_1\cup X_2)+d(X_1\cup X_2,j) 
&\ge 4(0.5\beta-2) \nonumber\\
&> 2\beta-M.
\end{align}
Finally, if a path from $i$ to $j$ passes through no path $P^k$ and no vertex $k$ for all $k\neq i,j$, then the path cannot visit any $v_k$-subgraph for $k\neq i,j$.
Thus, the path must go from $i$ through the $v_i$-subgraph to some $x\in X_1\cup X_2$, then through the $v_j$-subgraph to $j$.
If $x\in X_1$, the path has length
\begin{align}
  d(i,x_{X_1}) + d(x_{X_1},j)\ge (\beta+v_i[x]) + (\beta-v_j[x]) \ge 2\beta-M,
\end{align}
by assumption of $i$ and $j$, and similarly if $x\in X_2$ the path has length
\begin{align}
  d(i,x_{X_2}) + d(x_{X_2},j)\ge (\beta-v_i[x]) + (\beta+v_j[x]) \ge 2\beta-M,
\end{align}
as desired.
\end{proof}

\section{Unweighted Roundtrip $5/3-\varepsilon$ hardness from All-Nodes $k$-Cycle}
\label{sec:ankc-unw}

In this section, we extend the proof from Section~\ref{sec:ankc-unw} to unweighted graphs.

\begin{theorem*}[Theorem~\ref{thm:lb-ankc-unw}, restated]
  Under Hypothesis~\ref{hypo:ankc}, for all $\varepsilon,\delta>0$, no algorithm can $5/3-\varepsilon$ approximate the roundtrip diameter of a sparse directed unweighted graph in $O(n^{2-\delta})$ time.
\end{theorem*}

\begin{figure}
    \centering
    \includegraphics[width=0.95\textwidth]{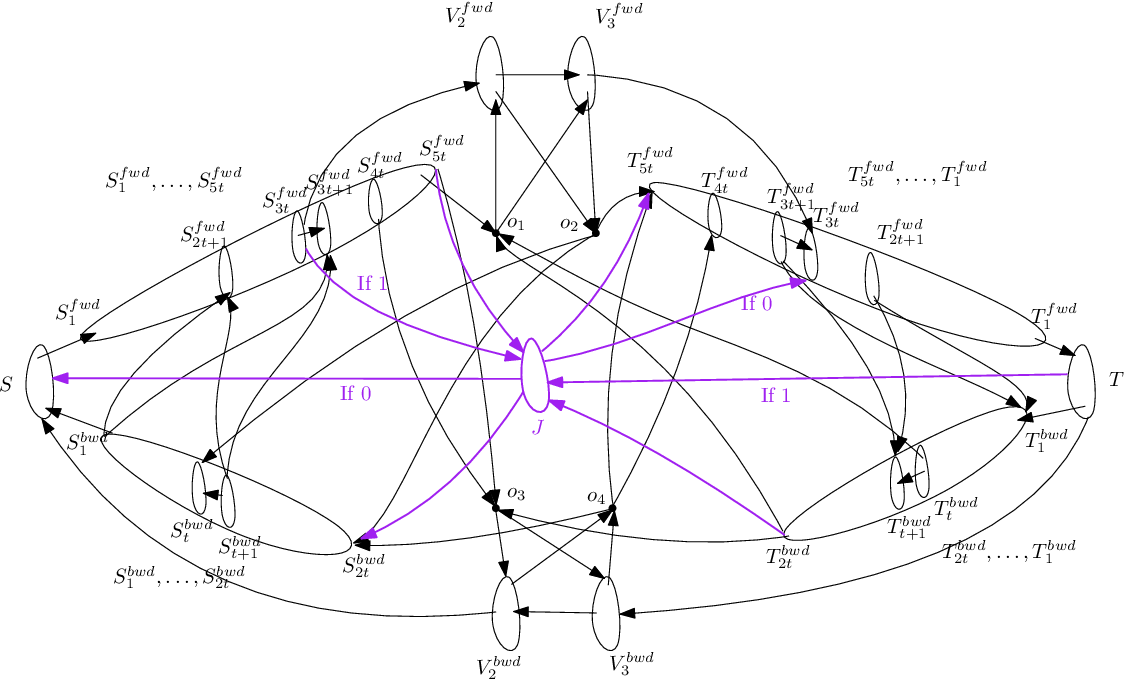}
    \caption{All-Nodes $3$-Cycle to Unweighted Roundtrip Diameter Approximation. The conditions on purple edges mean the following: if there is an edge between a node $a$ in the $S$-area or $T$-area and a a node $g_i\in J$, the condition If $\ell$ for $\ell=0,1$ mean that the edge exists if $\bar{a}[i] = \ell$. The copies of $S$ and $T$ are shown in long-narrow shaped vertex sets, and individual copies like $S_i^{fwd}$ that have specific edges to them are shown inside these sets. The edges inside the sets are not shown. Note that there are no edges between $o_i$ and $o_j$ for any $i,j\in\{1,2,3,4\}$ unlike the weighted case.}
    \label{fig:unweighted_reduction}
\end{figure}
We change the weighted construction as follows. We make $7t$ copies of $S$ and call them $S_{i}^{fwd}$ for $i=1,\ldots,5t$ (forward copies), and $S_i^{bwd}$ for $i=1,\ldots,2t$ (backward copies). Similarly we make $7t$ copies of $T$ and call them $T_i^{fwd}$ for $i=1,\ldots,5t$, and $T_i^{bwd}$ for $i=1,\ldots,2t$. These copies are the only new vertices added to the weighted construction. We call the subset of the graph containing $S$ and all its copies the $S$-area. Similarly, we call the subset containing $T$ and all its copies the $T$-area. We define the edges between these copies as follows. See Figure \ref{fig:unweighted_reduction}.

\begin{itemize}
    \item 
    We put perfect matchings between these copies. Formally, for $a\in V_1$, we add the edges $(a,a)\in S\times S^{fwd}_1$ and $(a,a)\in S^{fwd}_i\times S_{i+1}^{fwd}$ for $i=1,\ldots 5t-1$. We add the edges $(a,a)\in S_1^{bwd}\times S$ and $(a,a)\in S_{i+1}^{bwd}\times S_{i}^{bwd}$ for $i=1,\ldots,2t-1$. 
    \item For $a\in V_1$, we add the edges $(a',a')\in  T^{fwd}_1\times T$ and $(a',a')\in T^{fwd}_{i+1}\times T_{i}^{fwd}$ for $i=1,\ldots 5t-1$. We add the edges $(a',a')\in T\times T_1^{bwd}$ and $(a',a')\in T_i^{bwd}\times T_{i+1}^{bwd}$ for $i=1,\ldots,2t-1$. 
\end{itemize}
Note that so far we have a path of length $5t+1$ out of each $a\in S$ and a path of length $2t+1$ to each $a\in S$. Similarly we have a path of length $5t+1$ to each $a'\in T$ and a path of length $2t+1$ from each $a'\in T$. Now we can add edges that simulate the edges in the weighted construction. We start by defining the edges adjacent to $V_i^{fwd}$ and $V_i^{bwd}$ for $i=2,\ldots,k$. Note that the edges in $V_i^{fwd}\times V_{i+1}^{fwd}$ and $V_{i+1}^{bwd}\times V_i^{bwd}$ for $i=2,\ldots,k-1$, and the edges in $(V_i^{fwd}\cup V_i^{bwd})\times \{o_1,\ldots,o_4\}$ for $i=2,\ldots,k$ are the same as the weighted case and we include them here for completeness. 
\begin{itemize}
\item For all $a\in V_1$ and $x\in V_2$, add the edge $(a,x^{fwd})\in S_{3t}^{fwd}\times V_2^{fwd}$ if $(a,x)\in E(G)$. Add the edge $(x^{bwd},a)\in V_2^{bwd}\times S$ if $(x,a)\in E(G)$. 
\item Similarly, for any $a\in V_1$ and $x\in V_k$, add the edge $(x^{fwd},a')\in V_k^{fwd}\times T_{3t}^{fwd}$ if $(x,a)\in E(G)$. Add the edge $(a',x^{bwd})\in T\times V_k^{bwd}$ if $(a,x)\in E(G)$. 
\item The following edges are the same as in the weighted case and we note them for completeness. For each $ i \in \{2,\ldots,k\}$ and for each edge $(x,y) \in V_i \times V_{i+1}$ in $G$, we add two edges to $G'$: one forwards $(x^{fwd},y^{fwd}) \in V_i^{fwd} \times V_{i+1}^{fwd}$ and one backwards $(y^{bwd},x^{bwd}) \in V_{i+1}^{bwd} \times V_{i}^{bwd}$. The weight on these edges is $1$, which can be thought of as negligible because it is $0\cdot t +1$.
\end{itemize}

We define the edges adjacent to $o_i$ for $i=1,\ldots,4$.
\begin{itemize}
    \item For all $a\in V_1$, we add $(a,o_1)\in S_{5t}^{fwd}\times o_1, (a,o_3)\in (S_{4t}^{fwd}\cup S_{5t}^{fwd})\times o_3$ and $(o_2,a)\in o_2\times(S_{2t}^{bwd}\cup S_t^{bwd})$\footnote{Note that if we want to copy the weighted case, intuitively we should add edges $S_{4t}^{fwd}\times o_3$. Adding edges from $S_{5t}^{fwd}$ to $o_3$ only makes longer paths from $S$ so wouldn't hurt the yes case.}. 
    \item For all $a\in V_1$, we add $(o_2,a')\in o_2\times T_{5t}^{fwd}, (o_4,a')\in o_4\times( T_{4t}^{fwd}\cup T_{5t}^{fwd})$ and $(a',o_1)\in (T_{2t}^{bwd}\cup T_t^{bwd})\times o_1$.
\end{itemize}
The following edges exists in the weighted version as well, and we put them here for completeness. Note that there are no edges between any $o_i$ and $o_j$ in the unweighted case.
\begin{itemize}
    \item Add edges from $o_1$ to all nodes $v \in V^{fwd}_2 \cup \cdots \cup V^{fwd}_k$. Add an edge from all $v \in V^{fwd}_2 \cup \cdots \cup V^{fwd}_k$ to $o_2$
    \item Add edges from $o_3$ to all nodes $v \in V^{bwd}_2 \cup \cdots \cup V^{bkw}_k$. Add an edge from all $v \in V^{bwd}_2 \cup \cdots \cup V^{bwd}_k$ to $o_4$
\end{itemize}

Now we add edges adjacent to $J$.
\begin{itemize}
    \item For all $a\in V_1$ and $j\in [d]$, we add the edge $(a,g_j)\in S_{5t}^{fwd}\times J$ and we add $(g_j,a)\in J\times S_{2t}^{bwd}$\footnote{In the weighted case the $S_{5t}^{fwd}\times J$ edges have the constraint $a[j]=0$, but if we drop this condition it wouldn't hurt the yes case.}. If $\bar{a}[j]=1$, we add the edge $(a,g_j)\in S_{3t}^{fwd}\times J$. If $\bar{a}[j]=0$, we add the edge $(g_j,a)\in J\times S$.
    \item For any $a\in V_1$ and $j\in [d]$, we add the edge $(g_j,a')\in J\times T_{5t}^{fwd}$ and we add $(a',g_j)\in  T_{2t}^{bwd}\times J$. If $\bar{a}[j]=0$, we add the edge $(g_j,a')\in J\times T_{3t}^{fwd}$. If $\bar{a}[j]=1$, we add the edge $(a,g_j)\in T\times J$.
\end{itemize}

Finally we add the following edges that don't simulate any edges in the weighted case, and their use is to make copies of $S$ (and $T$) close to each other.
 
\begin{itemize}
    \item For all $a\in V_1$, add edges $(a,a)\in (S_1^{bwd}\cup S_{t+1}^{bwd})\times (S_{2t+1}^{fwd}\cup S_{3t+1}^{fwd})$.
    \item For all $a\in V_1$, add edges $(a',a')\in (T_{2t+1}^{fwd}\cup T_{3t+1}^{fwd})\times (T_1^{bwd}\cup T_{t+1}^{bwd}) $.
\end{itemize}

Note that we do not add any edges between $o_i$s, or between $o_i$ and $J$, as the existence of the above edges make it unnecessary.

\paragraph*{NO case} 
We compute distances with an $O(1)$ additive error to make the proof simpler.

First we cover the roundtrip distance between nodes that are in the $S$-area and $T$-area. Let $a,b\in S$-area. Let $J_i(a)=\{g_j|a[j]=i\}$ for $i\in\{0,1\}$, and let $S_0^{fwd}=S$ and $T_0^{fwd}=T$. 

\begin{lemma}
\label{lem:S-T-closeness}
Let $a,b\in S-area\cup T-area$ such that $\bar{a}\neq \bar{b}$, where $\bar{a}$ and $\bar{b}$ are $a$ and 
$b$'s identifiers. Suppose that $d(a,J)+d(J,a)\le 3t$ and $d(b,J)+d(J,b)\le 3t$. Then $d_{rt}(a,b)\le 6t$.
\end{lemma}

\begin{proof}
By the construction of the graph, we know that there is $i_1\in \{0,1\}$ such that $d(a,J)=d(a,g_j)$ for all $g_j\in J_{i_1}(a)$. Similarly there exist $i_2,i_3,i_4\in \{0,1\}$ such that 
\begin{itemize}
    \item $d(J,a)=d(g_j,a)$ for all $g_j\in J_{i_2}(a)$
    \item $d(b,J)=d(b,g_j)$ for all $g_j\in J_{i_3}(a)$
    \item $d(J,b)=d(g_j,b)$ for all $g_j\in J_{i_4}(a)$
\end{itemize}
Now since $\bar{a}\neq \bar{b}$, There is $g_j\in J_{i_1}(a)\cap J_{i_4}(b)$. Similarly, there is $g_{j'}\in J_{i_2}(a)\cap J_{i_3}(b)$. So $d(a,b)\le d(a,g_j)+d(g_j,b)=d(a,J)+d(J,b)$ and $d(b,a)\le d(b,g_{j'})+d(g_{j'},a)=d(b,J)+d(J,a)$. So $d_{rt}(a,b)\le 6t$.
\end{proof}

Now we show that for all $a,b\in S-area\cup T-area$ where $\bar{a}\neq \bar{b}$, the conditions of Lemma \ref{lem:S-T-closeness} hold, and hence $d_{rt}(a,b)\le 6t$. Suppose $a\in S-area$. We have
\begin{itemize}
    \item If $a\in S_{i}^{fwd}$ for $i=0,\ldots,3t$, then $d(a,J)\le 3t-i$ and $d(J,a)\le 1+i$.
    \item If $a\in S_{3t+i}^{fwd}$ for $i=1,\ldots,2t$, then $d(a,J)\le 2t-i$ and $d(J,a)\le t+i$ using the edges in $J\times S_{2t}^{bwd}$ and $S_{t+1}^{bwd}\times S_{3t+1}^{fwd}$.
    \item If $a\in S_i^{bwd}$ for $i=1,\ldots, 2t$, then $d(a,J)\le i+t$ (through $S_{1}^{bwd}\times S_{2t+1}^{fwd}$ edges) and $d(J,a)\le 2t-i$.
\end{itemize}
Since $T-area$ is symmetric, we have similar results if $a\in T-area$. So we can apply Lemma \ref{lem:S-T-closeness}.

Now suppose that $a\in S-area$ and $a'\in T-area$ are from the same node $a\in V_1$. First suppose $a\in S_{3t+i}^{f}$ for some $i\in\{1,\ldots,2t\}$. We know that there exist $j,j'$ such that $d(b,g_j)+d(g_{j'},b)\le 3t$. Then since all edges in $S^{fwd}_{5t}\times J$ and $J\times S_{2t}^{bwd}$ exist, we have $d(g_j,a)\le t+i$ (Using $S_{t+1}^{bwd}\times S_{3t+1}^{fwd}$ edges) and $d(a,g_{j'})\le 2t-i$. So $d_{rt}(a,b)\le 6t$. If $b\in T_{3t+i}^{fwd}$ for $i\in \{1,\ldots,2t\}$, we have a symmetric argument. 

So suppose $a\notin S_{3t+1}^{fwd}\cup \ldots \cup S_{5t}^{fwd}$. Let the cycle passing through $a$ in $G$ be $ax_2\ldots x_k$ where $x_i\in V_i$ for $i=2,\ldots,k$.
\begin{itemize}
    \item Let $a\in S_i^{fwd}$ and $a'\in T_j^{fwd}$ for some $i,j\in\{0,\ldots,3t\}$. Then consider the cycle passing through copies of $a$ in all $S_{\ell}^{fwd}$ for $\ell=0,\ldots,3t$, then going to $x_i^{fwd}\in V_i^{fwd}$ for $i=2,\ldots,k$, then to all copies of $a'$ in $T_{\ell}^{fwd}$ for $\ell=3t,\ldots, 0$, $x_i^{bwd}\in V_i^{bwd}$ for $i=k,\ldots,2$ and finally back to the copy of $a$ in $S$. This cycle passes through $a$ and $a'$ and is of length $6t$.
    \item Let $a\in S_i^{fwd}$ and $a'\in T_j^{bwd}$ for $i\in \{0,\ldots, 3t\}$ and $j\in \{1,\ldots, 2t\}$. Let $z\in [d]$ be a coordinate such that $a[z]=0$. The cycle passing through $a$ and $a'$ is the following: start from copies of $a$ in $S_{\ell}^{fwd}$ for all $\ell=0,\ldots,3t$, then to $x_i^{fwd}\in V_i^{fwd}$ for $i=2,\ldots,k$, to all copies of $a'$ in $T_{\ell}^{fwd}$ for $\ell=3t,\ldots,2t+1$, all copies of $a'$ in $T_{\ell}^{bwd}$ for $\ell=1,\ldots,2t$ then to $g_{z}$ and then back to the copy of $a$ in $S$.
    \item Let $a\in S_i^{bwd}$ and $a'\in T_j^{bwd}$ for some $i,j\in \{1,\ldots, 2t\}$. The cycle passing through $a$ and $a'$ is the following: start from copies of $a$ in $S_{\ell}^{fwd}$ for all $\ell=2t+1,\ldots 3t$, then to $x_i^{fwd}\in V_i^{fwd}$ for $i=2,\ldots,k$, to all copies of $a'$ in $T_{\ell}^{fwd}$ for $\ell=3t,\ldots,2t+1$, all copies of $a'$ in $T_{\ell}^{bwd}$ for $\ell=1,\ldots,2t$, then to $g_{z}$ for some arbitrary $z\in [d]$, to all the copies of $a$ in $S_{\ell}^{bwd}$ for $\ell=2t,\ldots,1$ and finally back to $S_{2t+1}^{fwd}$. 
\end{itemize}

Now we show that $S-area$ nodes are close to all nodes in $J$, $V_i^{fwd}$, $V_i^{bwd}$ and $o_j$ for $i=2,\ldots,k$ and $j=1,\ldots, 4$.

Let $a\in S-area$. We show that $d(a,o_1)+d(o_2,a)\le 6t$ and $d(a,o_3)+d(o_4,a)\le 6t$. Then since for every $x^{fwd}\in V_i^{fwd}$ for any $i\in \{2,\ldots,k\}$ there is a $2$-path $o_1x^{fwd}o_2$, and for every $x^{bwd}\in  V_i^{bwd}$ there is a $2$-path $o_3x^{bwd}o_4$, we have that $a$ is close to all nodes in $V_i^{fwd}\cup V_i^{bwd}\cup o_j$. The proof for $a\in T-area$ is similar. 

\begin{lemma}
For $a\in S-area$, we have that $d(a,o_1)+d(o_2,a)\le 6t$ and $d(a,o_3)+d(o_4,a)\le 6t$.
\end{lemma}
\begin{proof}
We do case analysis.
\begin{itemize}
    \item If $a\in S_i^{fwd}$ for $i=0,\ldots, 5t$, then $d(a,o_1)=5t-i$, and $d(o_2,a)=t+i$. 
    \item If $a\in S_i^{bwd}$ for $i=1,\ldots, 2t$, then $d(a,o_1)=i+3t+1$ using $S_1^{bwd}\times S_{3t+1}^{fwd}$ edges, and $d(o_2,a)=2t-i$.
    \item If $a\in S_i^{fwd}$ for $i=0,\ldots,4t$, then $d(a,o_3)=4t-i$ using edges in $S_{4t}^{fwd} \times o_3$ and $d(o_4,a)=2t+i$, using the edges in $o_4\times S_{2t}^{bwd}$.
    \item If $a\in S_{4t+i}^{fwd}$ for $i=1,\ldots,t$, then $d(a,o_3)=t-i$ and $d(o_4,a)=2t+i$ using the edges in $S_{t+1}^{bwd}\times S_{3t+1}^{fwd}$.
    \item If $a\in S_i^{bwd}$ for $i=1,\ldots, 2t$, then $d(a,o_3)= i+2t$ using $S_{1}^{bwd}\times S_{3t+1}^{fwd}$ edges, and $d(o_4,a)=2t-i$ using $o_4\times S_{2t}^{bwd}$ edges.
\end{itemize}
\end{proof}

Note that we can use these $6t$ paths from $o_2$ to $o_1$ in the Lemma to bound the roundtrip distances between $v^{fwd},u^{fwd}\in \cup_{i=2}^kV_i^{fwd}$ for any $v,u\in \cup_{i=2}^kV_i$. Similarly, we can use $6t$-paths from $o_4$ to $o_3$ to bound the roundtrip distances between $v^{bwd},u^{bwd}\in \cup_{i=2}^kV_i^{bwd}$, for any $v,u\in \cup_{i=2}^kV_i$.

Furthermore, we have that $d(o_2,o_3)=3t$ using $o_2\times T_{5t}^{fwd}$,  $T_{3t+1}^{fwd}\times T_{t+1}^{bwd}$ and $T_{2t}^{bwd}\times o_3$ edges. Symmetrically, $d(o_4,o_1)=3t$ using $o_4\times S_{2t}^{bwd}$, $S_{t+1}^{bwd}\times S_{3t+1}^{fwd}$ and $S_{5t}^{fwd}\times o_1$ edges. Now since for any $u\in \cup_{i=2}^kV_i$, $o_1u^{fwd}o_2$ and $o_3u^{bwd}o_4$ are paths of length $2$, using appropriate $2$-paths from $o_1$ to $o_2$ and from $o_3$ to $o_4$ we can form a cycle containing $x^{fwd}\in \cup_{i=2}^kV_i^{fwd}$ and $y^{bwd}\in \cup_{i=2}^kV_i^{bwd}$ for any $x,y\in \cup_{i=2}^kV_i$. 

Also note that using the above cycles, any $o_i$ for $i=1,\ldots,4$ and $x^{fwd}\in V_i^{fwd}$ or $x^{bwd}\in V_i^{bwd}$ are close for any $x\in V_i$ for $i=2,\ldots, k$.

Now it remains to prove that all the nodes in $J$ are close to all the other nodes. Fix some $g_j\in J$

\begin{itemize}
    \item For $a\in S_{3t+i}^{fwd}$ or $a\in S_i^{bwd}$ for some $i\in \{1,\ldots,2t\}$, there is a cycle passing through all copies of $a$ in $S_{3t+\ell}$ and $S_{\ell}^{bwd}$ for all $\ell=1,\ldots,2t$ and $g_j$, since there is an edge between all pairs in  $S_{5t}^{fwd}\times J$ and $J\times S_{2t}^{bwd}$.
    \item Let $a\in S_i^{fwd}$ for some $i\in \{0,\ldots, 3t\}$ and suppose $a[j]=1$. Then let $j'\in S$ be a coordinate where $a[j']=0$. Consider the following cycle: start from copies of $a$ in $S_{\ell}^{fwd}$ for all $\ell=0,\ldots,3t$, then go to $g_j$, then to copies of $a$ in $S_{\ell}^{bwd}$ for all $\ell=2t,\ldots,t+1$, then to copies of $a$ in $S_{3t+\ell}^{fwd}$ for all $\ell=1,\ldots,2t$, then to $g_{j'}$ and finally back to $S$.
    \item Let $a\in S_i^{fwd}$ for some $i\in \{0,\ldots, 3t\}$ and suppose $a[j]=0$. Then let $j'\in S$ be a coordinate where $a[j']=1$. We consider the same cycle above where we swap $g_j$ and $g_{j'}$ in the cycle. 
    \item To show that $J$ is close to $o_1,o_2,V_i^{fwd}$ for $i=2,\ldots,k$, we consider the following cycle. Let $a\in V_1$ be an arbitrary node. Start from copies of $a$ in $S_{3t+\ell}^{fwd}$ for all $\ell=1,\ldots,2t$, then go to $o_1$, then to a vertex in $V_i^{fwd}$ for some $i=2,\ldots,k$, then to $o_2$, to all copies of $a'$ in $T_{3t+\ell}^{fwd}$ for all $\ell=2t\ldots,1$, then all copies of $a'$ in $T_{\ell}^{bwd}$ for $\ell=t+1,\ldots,2t$, to $g_j$ then to all copies of $a$ in $S_{\ell}^{bwd}$ for $\ell=2t,\ldots,t+1$, and finally back to $S_{3t+1}^{fwd}$.
    \item To show that $J$ is close to $o_3,o_4,V_i^{bwd}$ for $i=2,\ldots,k$, we change the previous cycle. To go from the copy $a$ in $S_{5t}^{fwd}$ to the copy of $a$ in $T_{5t}^{bwd}$, we go through $g_j$. Then to go from the copy of $a$ in $T_{2t}^{bwd}$ to the copy of $a$ in $S_{2t}^{bwd}$, we go to $o_3$, then to any node in $V_i^{bwd}$ for some $i=2,\ldots,k$, then to $o_4$ and finally to $S_{2t}^{bwd}$.
    \item Suppose that we want to show $g_j$ is close to $g_{j'}$. Let $a$ be any node in $S$. Note that $d(g_j,g_{j'})=3t$ by going through all the copies of $a$ in $S_{\ell}^{bwd}$ for $\ell=2t,\ldots,t+1$, and the copies of $a$ in $S_{3t+\ell}^{fwd}$ for $\ell=1,\ldots,2t$. Similarly, $d(g_{j'},g_j)=3t$.
\end{itemize}

\paragraph*{YES case} 
In order to simplify the proof of the YES case, we note that the main difference between weighted and unweighted case that might cause short paths in the YES case are the edges added in $(S_{t+1}^{bwd}\cup S_1^{bwd})\times (S_{2t+1}^{fwd}\cup S_{3t+1}^{fwd})$ in the $S$-area (and the symmetric case in the T-area). We will show that if the roundtrip cycle uses any of these edges, the path is going to be long. If the path doesn't use any of these edges, then it is easy to see from the construction that there is an equivalent path in the weighted case.

We show that if $a\in S$ and $a'\in T$, $d(a,a')\ge 8t$ and $d(a',a)\ge 2t$.

First consider the $aa'$ path. The first $3t$ nodes on the path must be copies of $a$ in $S_i^{fwd}$ for $i=0,\ldots,3t$. Similarly, the last three nodes on this path must be copies of $a'$ in $T_i^{fwd}$ for $i=3t,\ldots,0$.

\begin{lemma}
Let $a\in S$ and $a'\in T$ be copies of the same node in $V_1$. If the $aa'$ shortest path uses any of the edges in $(S_{t+1}^{bwd}\cup S_1^{bwd})\times (S_{2t+1}^{fwd}\cup S_{3t+1}^{fwd})$, then this path has length at least $8t$.
\end{lemma}
\begin{proof}
First we show that $d(S,S_{t+1}^{bwd})\ge 4t$: This is because any path from $S$ to $S_{t+1}^{bwd}$ ends with nodes in $S_{i}^{bwd}$ for all $i=2t,\ldots,t+1$. Since it must go through $S_i^{fwd}$ for all $i=0,\ldots,3t$, we have $d(S,S_{t+1}^{bwd})\ge 4t$. 
Similar as above, we have that $d(S,S_{1}^{bwd})\ge 4t$. Now note that the distance from $S_{2t+1}^{fwd}\cup S_{3t+1}^{fwd}$ to any edge going out of the $S$-area is at least $t$. So before entering the $3t$-subpath in the $T$-area that ends in $a'$, the path has length $\ge 4t+t=5t$. So in total it has length at least $8t$.
\end{proof}

With a symmetric argument we can show that if the path $aa'$ uses any edge in $(T_{3t+1}^{fwd}\cup T_{2t+1}^{fwd}) \times (T_{1}^{bwd}\cup T_{t+1}^{bwd})$, it has length at least $8t$.

For the $a'a$ path, it is easier to see that if the path uses any of these edges, the length of it is at least $2t$. This is because $d(T,T_{3t+1}^{fwd}\cup T_{2t+1}^{fwd})>2t$ and $d(S_{3t+1}^{fwd}\cup S_{2t+1}^{fwd}, S)>2t$.

\end{document}